\RequirePackage{fix-cm}
\documentclass[aps,prx,reprint,superscriptaddress,longbibliography,nofootinbib,floatfix]{revtex4-2}

\usepackage{amsmath,amssymb,mathtools,bm}
\usepackage{amsthm}
\usepackage[T1]{fontenc}
\usepackage{newtxtext,newtxmath}
\usepackage{graphicx}
\usepackage{quantikz}
\usepackage{microtype}
\usepackage{placeins}
\usepackage[colorlinks=true,linkcolor=blue,citecolor=blue,urlcolor=blue]{hyperref}
\hypersetup{
  pdftitle={Entangling Topological Invariants},
  pdfauthor={Kazuki Ikeda and Yaron Oz},
  pdfsubject={Subsystem-gluing quotients and entangling topology of occupied-state bundles},
  pdfkeywords={Chern topology, Chern number, non-Abelian geometry, subsystem structure, topological pumping}
}

\newtheorem{theorem}{Theorem}
\newtheorem{proposition}{Proposition}
\newcommand{\ii}{\mathrm{i}}
\newcommand{\dd}{\mathrm{d}}
\newcommand{\Tr}{\operatorname{Tr}}
\newcommand{\cF}{\mathcal F}
\newcommand{\cA}{\mathcal A}
\newenvironment{widefigure}{%
  \par\onecolumngrid\vspace{6pt}%
}{%
  \par\vspace{6pt}\twocolumngrid%
}

\begin{document}
\raggedbottom

\title{Entangling Topological Invariants}
\author{Kazuki Ikeda}
\email{Kazuki.Ikeda@umb.edu}
\affiliation{Department of Physics, University of Massachusetts, Boston, MA 02125, USA}
\affiliation{Center for Nuclear Theory, Department of Physics and Astronomy, Stony Brook University, Stony Brook, NY 11794, USA}
\author{Yaron Oz}
\email{yaronoz@tauex.tau.ac.il}
\affiliation{School of Physics and Astronomy, Tel-Aviv University, Tel-Aviv 69978, Israel}

\begin{abstract}
An isolated occupied multiplet may admit local tensor-product descriptions without a globally consistent subsystem structure. We characterize the obstruction by comparing the transition functions of the occupied multiplet with those generated by independent basis changes in the two candidate subsystems. When a decomposition into rank-one sectors over a closed surface is specified, the resulting quotient removes row- and column-additive Chern data and yields mixed Chern classes. Momentum-dependent mixing of the sector labels adds the Gauss--Codazzi curvature of the moving lines, while in the label-conserving limit the mixed class is measured by a crossed Thouless pump. When only the factor dimensions $p$ and $q$ are specified, the comparison is made at the level of the clutching map of a rank-$pq$ bundle over $S^4$. Product frames generate winding numbers in $q\mathbb Z+p\mathbb Z$, so global factorization is possible exactly when $C_2$ is divisible by $\gcd(p,q)$; in particular, odd $C_2$ obstructs a $2\times2$ factorization. We illustrate the two settings with finite eight-level Hamiltonians and give pumping and occupied-projector tomography protocols for their readout.
\end{abstract}

\maketitle

\section{Introduction}

Subsystem labels such as spin, valley, layer, and orbital specify locally addressable degrees of freedom in a multiband system. Locally, a proposed $p\times q$ subsystem structure is an identification $E_x\simeq\mathbb C^p\otimes\mathbb C^q$, defined up to independent basis changes on the two factors. It extends globally only when the transition functions between local occupied-state frames can be chosen in the image of $U(p)\times U(q)\to U(pq)$. A smooth, gapped occupied multiplet need not satisfy this condition.

The obstruction is obtained by comparing occupied-bundle gluing with gluing induced by independent subsystem-frame changes. The relevant topological data depend on how much of the local decomposition is specified. A line resolution gives first-Chern data for the sectors; specifying only the occupied rank and the candidate factor dimensions leaves the clutching class of the full frame bundle.

On a closed surface, spin- and valley-Chern numbers, projected-observable constructions, and Chern-number matrices retain sector-resolved responses \cite{Sheng2006,Prodan2009,Ezawa2014,Wang2023Feature,Hung2026Nested,Niu1985,ZengShengZhu2019,ZengZhu2022}. We remove every contribution generated by either label separately. For binary labels, the surviving mixed class is the alternating sum of four sector Chern numbers. If the sector lines rotate inside a non-Abelian occupied space, its local curvature includes a Gauss--Codazzi term. In the symmetry-preserving limit, flux insertion for one label pumps the other label charge and measures the same mixed class \cite{Thouless1983}.

On $S^4$ the factorization question can be posed without a line resolution. A rank-$pq$ bundle is specified by an equatorial transition map in $SU(pq)$, while product frames generate the image of $SU(p)\times SU(q)$. The induced map on winding numbers has image $q\mathbb Z+p\mathbb Z$. We prove that the residue in $\mathbb Z_{\gcd(p,q)}$ is the complete obstruction to global factorization on $S^4$. For $p=q=2$, it is the parity of $C_2$. This connects the topology of occupied-state bundles with tensor-product reductions \cite{Ershov2003,Ershov2008}; related subsystem geometry and loop-dependent holonomy appear in Refs.~\cite{IkedaEntanglementGeometry,IkedaOzHolonomy}.

The paper is organized as follows. Section~\ref{sec:gluing-quotient} introduces the subsystem-gluing quotient used throughout. Section~\ref{sec:mixed-chern} develops its line-resolved form, including the mixed Chern invariant, the curvature of moving sector frames, and the crossed pump. Section~\ref{sec:full} treats the clutching problem on $S^4$, proves the factorization criterion, and presents the eight-level realization and transition-function tomography. Section~\ref{sec:discussion} discusses the physical interpretation and scope of the construction. Appendices~\ref{app:resolved-mixed}--\ref{app:eight-level} give the algebraic derivations, pumping calculations, proof details, and the three-transmon compilation and finite-shot analysis.

\section{Subsystem-gluing quotient}
\label{sec:gluing-quotient}

Let $E\to X$ be an isolated occupied bundle. On a patch $U_\alpha$, choose a local identification of its fibers with $\mathbb C^p\otimes\mathbb C^q$. Write $G_A$ and $G_B$ for the frame groups retained on the two candidate factors and $G_{AB}$ for the occupied-frame group. Independent changes of subsystem frame act through
\begin{equation}
\iota:G_A\times G_B\longrightarrow G_{AB},
\qquad
\iota(g_A,g_B)=g_A\otimes g_B.
\label{eq:subsystem-frame-map}
\end{equation}
In each of the two problems below, the relevant Chern or homotopy data form Abelian groups. We denote the problem-specific groups by $\mathcal G_A(X)$, $\mathcal G_B(X)$, and $\mathcal G_{AB}(X)$. The induced homomorphism $\iota_*:\mathcal G_A(X)\oplus\mathcal G_B(X)\to\mathcal G_{AB}(X)$ then defines
\begin{equation}
\mathfrak E_{A|B}(X)
:=\operatorname{coker}\iota_*
=\frac{\mathcal G_{AB}(X)}{\operatorname{Im}\iota_*}.
\label{eq:subsystem-gluing-cokernel}
\end{equation}
Equation~\eqref{eq:subsystem-gluing-cokernel} is an organizing form for the two concrete Abelian classification problems considered here: it removes gluing classes generated independently on the candidate factors.

For a prescribed array of rank-one sectors on a surface, the map acts on first-Chern data as
\begin{equation}
\Phi:\mathbb Z^{N_A}\oplus\mathbb Z^{N_B}
\longrightarrow\mathbb Z^{N_A\times N_B},
\qquad
\Phi(u,v)_{ab}=u_a+v_b.
\label{eq:surface-gluing-map}
\end{equation}
For a rank-$pq$ bundle over $S^4$, clutching identifies the relevant groups with third homotopy groups, and the tensor-product map becomes
\begin{equation}
\Psi_{p,q}:\mathbb Z\oplus\mathbb Z\longrightarrow\mathbb Z,
\qquad
\Psi_{p,q}(m,n)=qm+pn.
\label{eq:sphere-gluing-map}
\end{equation}
Consequently,
\begin{equation}
\operatorname{coker}\Phi\simeq\mathbb Z^{(N_A-1)(N_B-1)},
\qquad
\operatorname{coker}\Psi_{p,q}\simeq\mathbb Z_{\gcd(p,q)}.
\label{eq:two-gluing-cokernels}
\end{equation}
The two cokernels retain different levels of information about the candidate subsystem structure. With spectral lines specified, the quotient keeps the mixed combinations that cannot be generated by row and column data. With only the factor dimensions specified, it reduces the clutching integer by the subgroup generated by product frames. We begin with the surface case, where the quotient has a local curvature density and a transport interpretation.

\section{Mixed Chern topology from the surface quotient}
\label{sec:mixed-chern}

\subsection{Sector-line cokernel}
Let $X$ be a closed oriented surface and suppose the occupied bundle is resolved into rank-one sector bundles,
\begin{equation}
E=\bigoplus_{a=1}^{N_A}\bigoplus_{b=1}^{N_B}L_{ab}.
\label{eq:resolved-decomposition}
\end{equation}
This resolution makes the map $\Phi$ in Eq.~\eqref{eq:surface-gluing-map} concrete. A product frame has transition phases
\begin{equation}
g_{ab}=g_a^A g_b^B.
\label{eq:resolved-transition-factorization}
\end{equation}
Each sector Chern number $C_{ab}=c_1(L_{ab})[X]$ is an absolute invariant of the chosen decomposition. The mixed invariants are the integer combinations of $C_{ab}$ that vanish on every table of the form $u_a+v_b$. Equivalently, the sector Chern table defines a class
\begin{equation}
[C]_{AB}\in
\frac{\mathbb Z^{N_A\times N_B}}
{\{(u_a+v_b)_{ab}:u\in\mathbb Z^{N_A},\ v\in\mathbb Z^{N_B}\}}.
\label{eq:mixed-quotient}
\end{equation}
An $A$-only contribution shifts every entry in row $a$ by the same integer $u_a$, while a $B$-only contribution shifts every entry in column $b$ by $v_b$. The quotient extracts the part of the sector Chern table that is insensitive to additive $A$-only and $B$-only terms. For rank-one sectors it has rank $(N_A-1)(N_B-1)$,
and adjacent mixed differences provide integer coordinates,
\begin{equation}
\chi_{ab}=C_{ab}-C_{a,b+1}-C_{a+1,b}+C_{a+1,b+1}.
\label{eq:mixed-difference}
\end{equation}
Equivalently, $\chi_{ab}=c_1(\mathcal L_{ab})[X]$ for
$\mathcal L_{ab}=L_{ab}\otimes L_{a,b+1}^{-1}\otimes L_{a+1,b}^{-1}\otimes L_{a+1,b+1}$. A proof is given in Appendix~\ref{app:mixed-quotient}. The classes $\chi_{ab}$ describe the determinant-line topology of the prescribed decomposition.

Let $\tau_i$ and $\mu_i$ act on the two binary label spaces and let $\sigma_i$ act on the two-component orbital. Consider the Qi--Wu--Zhang block \cite{QiWuZhang2006}
\begin{equation}
h(\bm k;M)=\sin k_x\,\sigma_x+\sin k_y\,\sigma_y+
(M+\cos k_x+\cos k_y)\sigma_z
\label{eq:qwz}
\end{equation}
and the eight-level Hamiltonian
\begin{equation}
H_0(\bm k;m)=\bigoplus_{a,b=\pm1}h(\bm k;abm).
\label{eq:four-block}
\end{equation}
The occupied lines are resolved by binary labels $Z_A=\tau_z$ and $Z_B=\mu_z$. We normalize their compact charges to $a,b=\pm1$, so a $2\pi$ label flux is the fundamental large gauge transformation. Define
\begin{align}
C_{\rm tot}&=\sum_{ab}C_{ab}, &
C_A&=\sum_{ab}aC_{ab},\nonumber\\
C_B&=\sum_{ab}bC_{ab}, &
\chi&=\sum_{ab}abC_{ab}.
\label{eq:binary-responses}
\end{align}
For $0<m<2$,
\begin{equation}
C=\begin{pmatrix}1&-1\\-1&1\end{pmatrix},
\qquad C_{\rm tot}=C_A=C_B=0,
\qquad \chi=4.
\label{eq:zero-marginal-table}
\end{equation}
Thus all total and one-label Chern responses vanish while the mixed Chern number remains nonzero [Fig.~\ref{fig:zero-marginal}].

\begin{figure*}[t]
\centering
\includegraphics[width=0.99\textwidth]{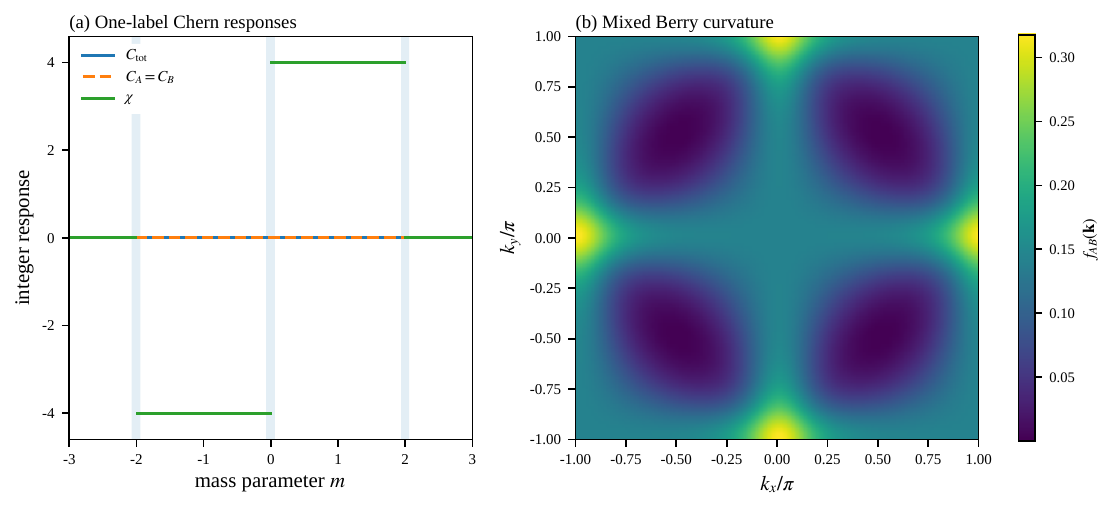}
\caption{\textbf{Mixed Chern topology with vanishing one-label responses.}
(a) Total and one-label responses vanish on the central plateaus, while $\chi=\pm4$. Shading marks bulk gap closings. (b) The mixed Berry curvature at $m=1$ integrates to four; the sector table is given in Eq.~\eqref{eq:zero-marginal-table}.}
\label{fig:zero-marginal}
\end{figure*}

\subsection{Moving sector frames}
\label{sec:mixing-pump}

To follow the surface quotient away from exact label conservation, we use the four-channel momentum-dependent family
\begin{align}
V(\bm k;\bm v)=&\;v_A\sin k_x\,\tau_x\sigma_z
+v_B\sin k_y\,\mu_x\sigma_x\nonumber\\
&+v_{AB}^{(\sigma)}(\cos k_x-\cos k_y)\tau_y\mu_y\sigma_y\nonumber\\
&+v_{AB}^{(0)}\sin(k_x+k_y)\tau_x\mu_x.
\label{eq:label-mixing}
\end{align}
Figures~\ref{fig:label-mixing} and~\ref{figS:label-resolution} follow the straight path
\begin{align}
H_\epsilon(\bm k)&=H_0(\bm k)+\epsilon V(\bm k;\bm v_\star),\nonumber\\
\bm v_\star&=\frac1{20}(11,10,7,4),
\qquad 0\le\epsilon\le\frac65.
\label{eq:mixing-path}
\end{align}
The first two channels rotate the individual labels, and the weaker joint channels mix them together. Their relative amplitudes remain fixed as $\epsilon$ controls the common scale. Because the harmonics depend on momentum, no constant basis change removes the deformation. The occupied energy gap remains open along the path. Let $P(\bm k)$ be the rank-four occupied spectral projector of $H_\epsilon(\bm k)$ and define
\begin{equation}
X_{\rm lab}(\bm k)=P(\bm k)(2Z_A+Z_B)P(\bm k).
\label{eq:projected-label-operator}
\end{equation}
We call the minimum adjacent eigenvalue spacing of $X_{\rm lab}|_{\operatorname{Ran}P}$ over the Brillouin zone the label-resolution gap. It remains open, so the four eigenlines can be tracked continuously from $\epsilon=0$. Figure~\ref{fig:label-mixing}(a) shows the two gaps, and Appendix~\ref{app:label-resolution} analyzes alternative compressed-label choices and their entanglement-spectrum diagnostics.

Let $Q_s(\bm k)$ be a resolved line projector, let $L_s=\operatorname{Ran}Q_s$, and let $\cF_P=P(\dd P)^{\wedge2}P$ be the occupied Wilczek--Zee curvature \cite{WilczekZee1984}. We include the factor $1/(2\pi\ii)$ so that the integral of $f_s$ is the sector Chern number:
\begin{align}
f_s&=\frac{1}{2\pi\ii}\left\{\Tr(Q_s\cF_P)\right.\left.+\Tr\!\left[Q_s(P\dd Q_sP)^{\wedge2}\right]\right\},\nonumber\\
\int_X f_s&=c_1(L_s)[X].
\label{eq:sector-curvature}
\end{align}
For weights $w_s=(1,-1,-1,1)$ in the order $(++,+-,-+,--)$,
\begin{align}
f_{AB}&=f_{\rm proj}+\Xi_{AB},\nonumber\\
f_{\rm proj}&=\frac{1}{2\pi\ii}\sum_s w_s\Tr(Q_s\cF_P),\nonumber\\
\Xi_{AB}&=\frac{1}{2\pi\ii}\sum_s w_s
\Tr\!\left[Q_s(P\dd Q_sP)^{\wedge2}\right].
\label{eq:curvature-decomposition}
\end{align}
The Gauss--Codazzi term $\Xi_{AB}$ records the motion of the resolved lines inside the occupied multiplet. The sum $f_{AB}$ is the determinant-line curvature; its two contributions depend on the moving sector frame. At the endpoint of Eq.~\eqref{eq:mixing-path}, the projected term supplies most of the integral and the Gauss--Codazzi term completes it to the Chern number four. Appendix~\ref{app:gauss-codazzi} gives the pointwise identity and mesh convergence. The endpoint has no conserved label charges, so the calculation is geometric. At $\epsilon=0$, the restored label symmetries give the mixed class a transport interpretation.

\begin{figure*}[t]
\centering
\includegraphics[width=0.99\textwidth]{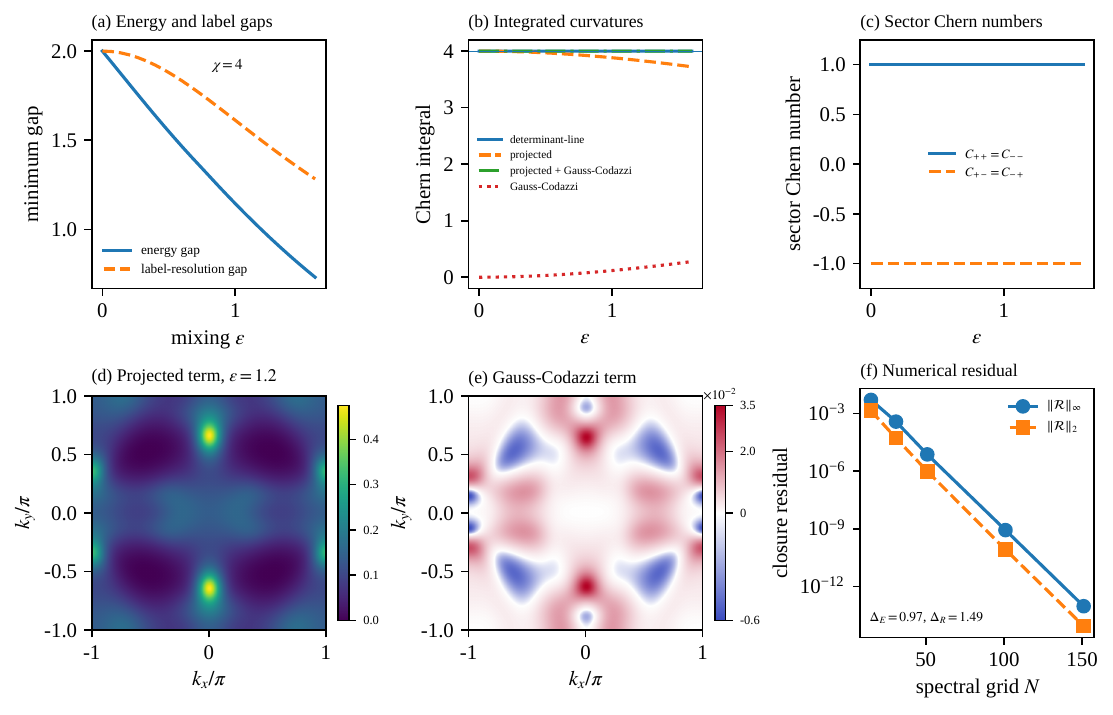}
\caption{\textbf{Mixed Chern topology under label mixing.}
(a) Energy and label-resolution gaps. (b) Determinant-line curvature and its projected and Gauss--Codazzi contributions. (c) Sector Chern numbers computed with the Fukui--Hatsugai--Suzuki algorithm \cite{FukuiHatsugaiSuzuki2005}. (d),(e) Projected and Gauss--Codazzi densities at $\epsilon=6/5$ on a momentum grid. Separate color scales display the two contributions; the Gauss--Codazzi scale is centered at zero. (f) Numerical residual of Eq.~\eqref{eq:curvature-decomposition}.}
\label{fig:label-mixing}
\end{figure*}

\subsection{Crossed pumping}

The cross response is evaluated in the $U(1)_A\times U(1)_B$-symmetric limit $H_0$. In the channel order $(++,+-,-+,--)$, the Abelian edge theory has chiral $K$ matrix and label-charge vectors
\begin{align}
K&=\operatorname{diag}(1,-1,-1,1),\nonumber\\
t_A&=(1,1,-1,-1)^T,\nonumber\\
t_B&=(1,-1,1,-1)^T.
\label{eq:charge-vectors}
\end{align}
They obey
\begin{equation}
t_A^TK^{-1}t_B=4=\chi,
\qquad
t_A^TK^{-1}t_A=t_B^TK^{-1}t_B=0.
\label{eq:mixed-anomaly}
\end{equation}
The null-vector criterion rules out a fully gapped, nondegenerate boundary that preserves both label charges \cite{Haldane1995}. Diagonalization of the eight-band ribbon shows that label-preserving edge potentials move the crossings, whereas a boundary $\tau_x$ term breaks $A$ conservation and opens a gap [Fig.~\ref{fig:edge-pump}(a,b)].

For a $2\pi$ flux coupled to $A$, sector $(a,b)$ pumps $aC_{ab}$ particles. The total, $A$-, and $B$-charge transfers are therefore
\begin{align}
\Delta Q_{\rm tot}&=\sum_{a,b}aC_{ab}=C_A,\nonumber\\
\Delta Q_A&=\sum_{a,b}a^2C_{ab}=C_{\rm tot},\nonumber\\
\Delta Q_B&=\sum_{a,b}abC_{ab}=\chi.
\label{eq:cross-pump-bulk}
\end{align}
For Eq.~\eqref{eq:zero-marginal-table}, this gives $\Delta Q_{\rm tot}=\Delta Q_A=0$ and $\Delta Q_B=4$.

We insert an $A$-dependent flux by solving
\begin{equation}
\ii\partial_t|\psi_j(t)\rangle=H[\theta_A(t)]|\psi_j(t)\rangle,
\qquad
\theta_A(t)=\theta_0+2\pi t/T,
\label{eq:time-pump}
\end{equation}
for every initially occupied state of a finite-cylinder Fermi sea. The evolved one-particle projector is
\begin{equation}
P_{\rm evo}(t)=\sum_{j\in\mathrm{occ}}|\psi_j(t)\rangle\langle\psi_j(t)|.
\label{eq:evolved-projector}
\end{equation}
For the cycle shown in Fig.~\ref{fig:edge-pump}, $\theta_0$ is chosen away from the exponentially small finite-size edge anticrossing. Let $\Pi_R$ be the one-particle projector onto the right half of the cylinder. The finite-cylinder calculation measures
\begin{align}
Q_{\rm tot}^R(t)&=\Tr[P_{\rm evo}(t)\Pi_R],\nonumber\\
Q_A^R(t)&=\Tr[P_{\rm evo}(t)\Pi_R Z_A],\nonumber\\
Q_B^R(t)&=\Tr[P_{\rm evo}(t)\Pi_R Z_B].
\label{eq:right-charge}
\end{align}
The cycle transfers four units of $B$ charge, with no net total or $A$-charge transfer. A finite bulk gap and an exponentially small opposite-edge anticrossing create a ramp window that follows the bulk adiabatically and crosses the edge anticrossing diabatically. Size, ramp-time, offset, and disorder checks are summarized in Appendix~\ref{app:pump}, together with the boundary-localized transfer profile.

\FloatBarrier
\begin{widefigure}
\begin{center}
\includegraphics[width=0.94\textwidth]{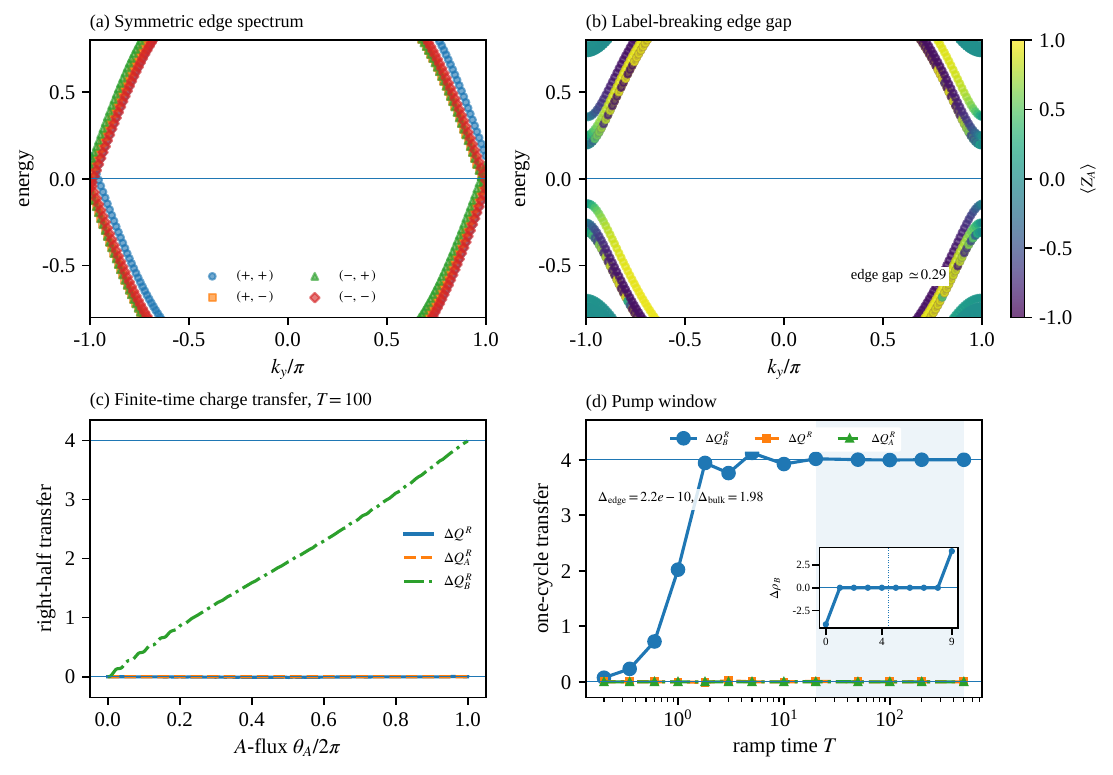}
\parbox{0.94\textwidth}{\small\refstepcounter{figure}\label{fig:edge-pump}\textbf{FIG.~\thefigure. Boundary spectrum and cross pump.}
(a) Edge states of the eight-band QWZ ribbon with both label charges conserved. Sector momentum offsets separate coincident crossings. (b) A boundary $\tau_x$ coupling opens a gap. (c) Finite-time evolution transfers only $B$ charge. (d) Finite-size and ramp-time dependence around the quantized transfer; the inset shows the boundary-localized density change.}
\end{center}
\end{widefigure}
\section{The clutching quotient on \texorpdfstring{$S^4$}{S4}}
\label{sec:full}

The surface calculation used a chosen eigenline resolution, so its gluing class appeared as a Chern table. Suppose now that no such resolution is supplied. The occupied rank and the candidate factor dimensions remain, and Eq.~\eqref{eq:subsystem-gluing-cokernel} acts on the clutching class of the full frame bundle. The resulting quotient is a global factorization test.

\subsection{Factorization from the clutching quotient}

Let $E\to S^4$ have rank $pq$. Since $H^2(S^4;\mathbb Z)=0$, its determinant line is topologically trivial. After choosing a determinant trivialization, independent subsystem operations generate the subgroup
\begin{align}
G_{p,q}&=\operatorname{Im}[SU(p)\times SU(q)\to SU(pq)],\nonumber\\
G_{p,q}&\simeq\frac{SU(p)\times SU(q)}{\mu_d},
\qquad d=\gcd(p,q).
\label{eq:local-group-su}
\end{align}
Here $\mu_d=\{(\zeta I_p,\zeta^{-1}I_q):\zeta^d=1\}$ is the common central kernel. A global $p\times q$ structure is a reduction of the $SU(pq)$ frame bundle to $G_{p,q}$. Write
\begin{equation}
C_2(E):=\langle c_2(E),[S^4]\rangle.
\label{eq:c2-definition}
\end{equation}
Writing $W_3$ for the integer winding of a map $S^3\to SU(n)$, clutching on the equatorial $S^3$ realizes Eq.~\eqref{eq:sphere-gluing-map} through
\begin{equation}
W_3[g_A\otimes g_B]=qW_3[g_A]+pW_3[g_B].
\label{eq:tensor-winding-main}
\end{equation}
Hence the quotient of occupied windings by subsystem-generated windings is
\begin{equation}
\frac{\mathbb Z}{q\mathbb Z+p\mathbb Z}
\simeq\mathbb Z_{\gcd(p,q)}.
\label{eq:s4-winding-cokernel}
\end{equation}
The next theorem shows that this residue is the complete obstruction to a global factorization on $S^4$, including the lift from the product subgroup to the two subsystem bundles.

\FloatBarrier
\begin{widefigure}
\begin{center}
\includegraphics[width=0.94\textwidth]{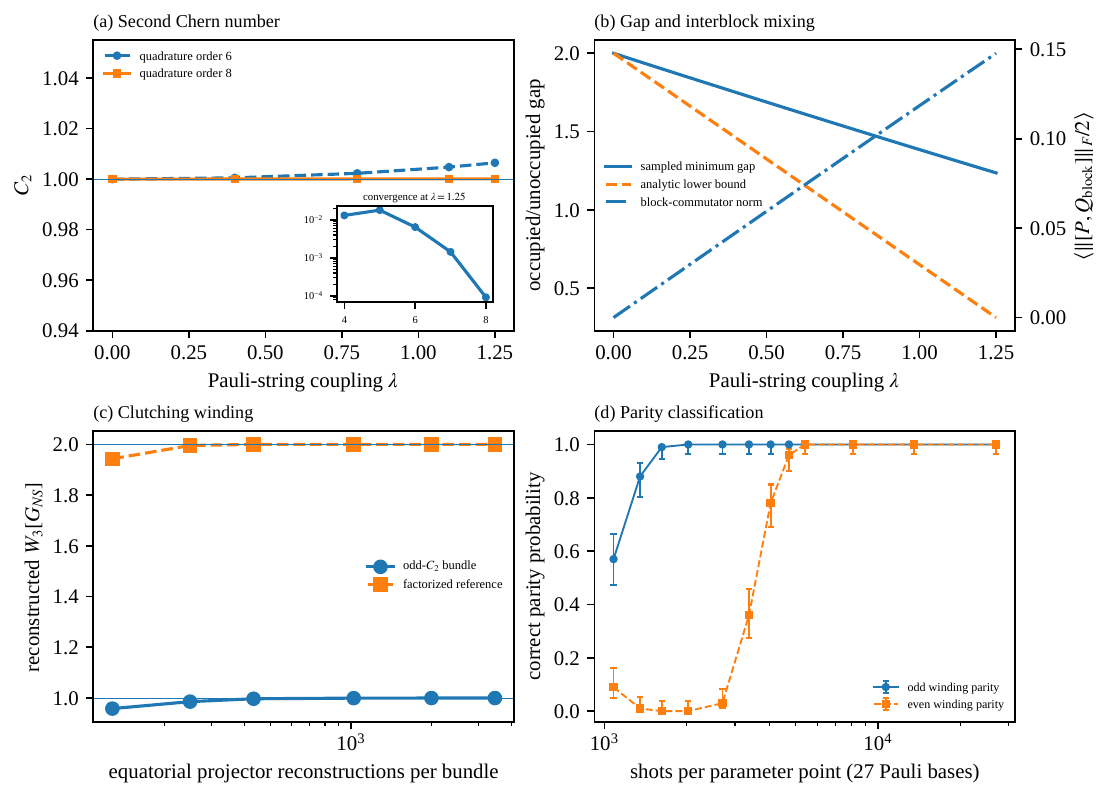}
\parbox{0.94\textwidth}{\small\refstepcounter{figure}\label{fig:full-reduction}\textbf{FIG.~\thefigure. Clutching quotient and subsystem-factorization test.}
(a) Numerical integration of the interblock-coupled rank-four projector converges to $C_2=1$. (b) The solid curve gives the minimum gap on the sampled sphere, and the dashed curve gives the analytic lower bound. The right axis shows the block-commutator norm $\tfrac12\langle\|[P_\lambda,Q_{\rm block}]\|_F\rangle$, where $Q_{\rm block}=\operatorname{diag}(I_4,0_4)$ projects onto the original Yang block. (c) The north--south clutching winding converges to one for the odd-$C_2$ bundle and two for the factorized reference. (d) Correct winding parity versus the finite-shot budget for the three-transmon encoding; error bars are $95\%$ Wilson intervals. The simulation protocol, gate compilation, and acquisition counts are given in Appendix~\ref{app:tomography}.}
\end{center}
\end{widefigure}

\begin{theorem}[Factorization criterion on $S^4$]\label{thm:factorization}
Let $p,q\ge2$, and let $E\to S^4$ be a rank-$pq$ complex vector bundle. For an ordered candidate factorization of dimensions $p\times q$, the following statements are equivalent:
\begin{enumerate}
\setlength{\itemsep}{2pt}
\setlength{\parskip}{0pt}
\setlength{\parsep}{0pt}
\item[(i)] After choosing a trivialization of $\det E$, the associated $SU(pq)$ frame bundle reduces to $G_{p,q}$.
\item[(ii)] There are rank-$p$ and rank-$q$ complex vector bundles $E_A,E_B$ with trivial determinant such that
\begin{equation*}
E\simeq E_A\otimes E_B.
\end{equation*}
\item[(iii)]
\begin{equation*}
C_2(E)\in \gcd(p,q)\mathbb Z.
\end{equation*}
\end{enumerate}
Hence, with $d=\gcd(p,q)$,
\begin{equation}
\eta_{p,q}(E):=C_2(E)\bmod d\in\mathbb Z_d
\label{eq:full-invariant}
\end{equation}
is the obstruction to factorization on $S^4$.
\end{theorem}

The residue is nontrivial when $p$ and $q$ share a common divisor; the smallest case is $p=q=2$. Appendix~\ref{app:factorization-proof} proves the converse construction and the lift through the finite central quotient. For two qubits,
\begin{equation}
C_2(E_A\otimes E_B)=2C_2(E_A)+2C_2(E_B)\in2\mathbb Z,
\label{eq:two-qubit-evenness}
\end{equation}
so odd $C_2$ forbids a global $2\times2$ factorization. The candidate subsystem dimensions enter the reduction map itself; changing them changes the image in Eq.~\eqref{eq:sphere-gluing-map} and therefore the quotient.

\subsection{Eight-level realization}

For the rank-four realization, the ambient eight-level control space is
\begin{equation}
\mathcal H_{\rm amb}=\mathbb C_\kappa^2\otimes\mathbb C_\tau^2\otimes\mathbb C_\sigma^2.
\label{eq:ambient-hardware}
\end{equation}
For a superconducting realization, identify $|\kappa\tau\sigma\rangle$ with the computational basis of three transmons $(q_\kappa,q_\tau,q_\sigma)$ arranged as the chain $q_\kappa$--$q_\tau$--$q_\sigma$. This physical tensor product supplies the control and readout basis; the theorem tests a possible two-factor structure inside the occupied rank-four bundle. Let the Yang-monopole negative-energy projector \cite{Yang1978,Sugawa2018} be
\begin{equation}
P_Y(\bm n)=\frac{I_4-\sum_{\mu=1}^{5}n_\mu\Gamma_\mu}{2},
\qquad \bm n\in S^4,
\label{eq:yang-projector}
\end{equation}
The $4\times4$ matrices $\Gamma_\mu$ obey $\{\Gamma_\mu,\Gamma_\nu\}=2\delta_{\mu\nu}$; an explicit representation is given in Appendix~\ref{app:eight-level}. Let $E_Y=\operatorname{Ran}P_Y$ be the rank-two occupied bundle. We orient $S^4$ so that $C_2(E_Y)=1$, which fixes the equatorial clutching winding to $W_3=1$. The odd projector and factorized reference are
\begin{align}
P_{\rm odd}&=P_Y\oplus I_2\oplus0_2,
& C_2&=1,& \eta_{2,2}&=1,\nonumber\\
P_{\rm ref}&=P_Y\otimes I_2,
& C_2&=2,& \eta_{2,2}&=0.
\label{eq:odd-even-projectors}
\end{align}
The factorized reference bundle is explicitly $E_Y\otimes\underline{\mathbb C}^{2}$, where $\underline{\mathbb C}^{2}$ is the trivial rank-two bundle. For the odd projector we define the spectrally flattened Hamiltonian
\begin{equation}
H_{\rm odd}^{\rm flat}=I_8-2P_{\rm odd}.
\label{eq:odd-flat-main}
\end{equation}
To break exact block conservation we study $H_\lambda=H_{\rm odd}^{\rm flat}+\lambda V_{\rm cpl}$ with
\begin{equation}
V_{\rm cpl}(\bm n)=\frac12\left(\frac{7}{20}M_0+\frac1{\sqrt5}\sum_{\mu=1}^{5}n_\mu M_\mu\right),
\label{eq:structured-s4-mixing}
\end{equation}
where
\begin{align}
(M_0,\ldots,M_5)=(&\kappa_x,\kappa_y\tau_x,\kappa_y\tau_z\sigma_x,\nonumber\\
&\kappa_x\tau_y\sigma_z,\kappa_z\tau_x\sigma_y,\tau_y\sigma_x).
\label{eq:structured-pauli-strings}
\end{align}
The Pauli weights of $(M_0,\ldots,M_5)$ are $(1,2,3,3,3,2)$. A weight-$w$ Pauli rotation is implemented exactly by local basis changes, parity accumulation on one qubit, and a single $R_z$ rotation, using $2(w-1)$ CNOT-equivalent entanglers. The six controls therefore use $(0,2,4,4,4,2)$ entanglers, respectively. Choosing $q_\tau$ as the parity target implements every weight-three term on the two edges of the linear chain, without SWAP gates. Figure~\ref{fig:transmon-circuits}(a) displays the resulting nearest-neighbor circuit for $M_2=YZX$; changing only the local basis rotations gives the other Pauli words. Panel (b) shows the state-preparation and local-basis readout circuit used to reconstruct the occupied projector. Appendix~\ref{app:eight-level} gives the general circuit identity and shows that the remaining terms in $H_{\rm odd}^{\rm flat}$ have the same maximal Pauli weight.

\FloatBarrier
\begin{widefigure}
\begin{center}
\begin{minipage}[t]{0.5\textwidth}
\centering
\textbf{(a) Nearest-neighbor implementation of $e^{-\ii\theta YZX/2}$}\par\vspace{4pt}
{\small
\begin{quantikz}[row sep=0.27cm,column sep=0.20cm]
\lstick{$q_\kappa$} & \gate{S^\dagger} & \gate{H} & \ctrl{1} & \qw & \qw & \qw & \ctrl{1} & \gate{H} & \gate{S} & \qw \\
\lstick{$q_\tau$}   & \qw & \qw & \targ{} & \targ{} & \gate{R_z(\theta)} & \targ{} & \targ{} & \qw & \qw & \qw \\
\lstick{$q_\sigma$} & \qw & \gate{H} & \qw & \ctrl{-1} & \qw & \ctrl{-1} & \qw & \gate{H} & \qw & \qw
\end{quantikz}}
\end{minipage}\hfill
\begin{minipage}[t]{0.5\textwidth}
\centering
\textbf{(b) One projector-tomography circuit}\par\vspace{4pt}
{\small
\begin{quantikz}[row sep=0.27cm,column sep=0.30cm]
\lstick{$\ket{0}_{q_\kappa}$} & \gate[wires=3]{U_j(\bm n)} & \gate{R_\alpha} & \meter{} \\
\lstick{$\ket{0}_{q_\tau}$}   & & \gate{R_\beta} & \meter{} \\
\lstick{$\ket{0}_{q_\sigma}$} & & \gate{R_\gamma} & \meter{}
\end{quantikz}}
\par\vspace{3pt}
{\small $j=1,\ldots,4$, \quad $(\alpha,\beta,\gamma)\in\{X,Y,Z\}^3$}
\end{minipage}
\parbox{0.94\textwidth}{\small\refstepcounter{figure}\label{fig:transmon-circuits}\textbf{FIG.~\thefigure. Quantum circuits.}
(a) Exact compilation of the weight-three control $M_2=YZX$. The first two and final two local gates on $q_\kappa$ implement $B_Y=HS^\dagger$ and $B_Y^\dagger=SH$; the Hadamards on $q_\sigma$ implement the $X$-basis change. The two CNOT pairs use only the device edges $q_\kappa$--$q_\tau$ and $q_\tau$--$q_\sigma$. (b) One circuit in the projector-tomography ensemble. $U_j(\bm n)$ prepares occupied eigenvector $j$, and $R_X=H$, $R_Y=HS^\dagger$, $R_Z=I$ rotate the selected Pauli basis to computational-basis readout. Counts from the four occupied preparations are pooled for each of the $27$ basis triples.}
\end{center}
\end{widefigure}

Since every $M_\mu$ has unit norm,
\begin{equation}
\sup_{\bm n}\|V_{\rm cpl}(\bm n)\|\le\frac{27}{40}\simeq0.68,
\qquad
\Delta_\lambda(\bm n)\ge2-\frac{27}{20}|\lambda|.
\label{eq:weyl-bound}
\end{equation}
At $\lambda=5/4$, the analytic bound remains positive, so the path is gapped and homotopy invariance gives $C_2=1$; direct quadrature verifies this value [Fig.~\ref{fig:full-reduction}(a)]. Let $P_\lambda$ denote the rank-four occupied spectral projector of $H_\lambda$. Spectral flattening, $\widetilde H_\lambda=I_8-2P_\lambda$, restores an exactly degenerate occupied quartet whenever a Wilczek--Zee interpretation is desired, with the bundle unchanged.

\subsection{Transition-function tomography}

To reconstruct the clutching map, let $P$ be the rank-four occupied projector and represent fixed north and south reference subspaces by four-column isometries $R_N$ and $R_S$. Define
\begin{equation}
U_\nu=P R_\nu(R_\nu^\dagger P R_\nu)^{-1/2},
\qquad
G_{NS}=U_N^\dagger U_S\in U(4).
\label{eq:tomographic-frame}
\end{equation}
The determinant component has trivial $\pi_3$, so the trace winding of $G_{NS}$ equals the winding of its special-unitary part. Frame changes that extend across the hemispheres leave this integer unchanged; Appendix~\ref{app:tomography} gives the proof and discretization. On the three-transmon register, each occupied eigenvector is compiled offline as a three-qubit preparation circuit, and equal numbers of shots are pooled to reconstruct $P/4$. Direct preparation avoids a product-formula sequence of Hamiltonian rotations; digital Pauli-string evolution remains available when the dynamics itself is of interest. For the $N=4$ grid and $200$ pooled shots per Pauli basis, the protocol uses approximately $7\times10^5$ shots per bundle. The reconstructed windings approach $1$ and $2$, and the finite-shot simulations resolve their parities.

\section{Discussion}
\label{sec:discussion}

The surface and $S^4$ invariants arise at different resolutions of the candidate subsystem structure. With spectral lines specified, the quotient removes row- and column-additive Chern data and leaves mixed curvature and transport. With only the factor dimensions specified, it reduces the clutching winding modulo the windings generated by product frames and tests global factorization.

For moving sector lines, the local density combines projected Wilczek--Zee curvature with the Gauss--Codazzi curvature of the sector frame. When the labels are conserved, flux insertion measures the mixed integer through crossed transport. On $S^4$, north--south frame overlaps reconstruct the transition function and its winding parity. Yang-monopole curvature and integrated second-Chern responses have been measured in atomic and photonic systems \cite{Sugawa2018,Lohse2018,Zilberberg2018}. The protocol developed here reconstructs the transition function itself. The three-transmon simulation tests the reconstruction from occupied-projector tomography to the parity of the clutching winding. The hardware qubits encode the ambient eight-dimensional control space; the factorization test concerns the occupied four-plane.

The construction extends to other base spaces and higher-rank spectral resolutions once the image of the subsystem-frame changes is determined. Curvature, transport, and transition-function measurements then provide complementary readouts of the resulting obstruction.

\appendix

\section{Resolved mixed topology}\label{app:resolved-mixed}
\label{secS:resolved-mixed}

\subsection{Mixed Chern quotient}\label{app:mixed-quotient}
\label{secS:mixed-quotient}

Each sector integer $C_{ab}=c_1(L_{ab})[X]$ is fixed by the specified decomposition. The quotient identifies Chern tables that differ by additive one-label contributions. Let $\Lambda=\mathbb Z^{N_A\times N_B}$ and define
\begin{align}
\Phi &: \mathbb Z^{N_A}\oplus\mathbb Z^{N_B}\longrightarrow\Lambda,\nonumber\\
\Phi(u,v)_{ab}&=u_a+v_b.
\label{eqS:Phi}
\end{align}
The sublattice $\Lambda_{\rm one}=\operatorname{im}\Phi$ contains the additive Chern tables attributable to either label separately. The mixed Chern datum is the quotient class $[C]_{AB}$ in $\mathcal Q_{AB}:=\operatorname{coker}\Phi=\Lambda/\Lambda_{\rm one}$. For example, tensoring $L_{ab}$ with $A_a\otimes B_b$ shifts the Chern table by $u_a+v_b$. Define the mixed-difference map
\begin{align}
D:\Lambda&\longrightarrow\mathbb Z^{(N_A-1)(N_B-1)},\nonumber\\
(DC)_{ab}&=C_{ab}-C_{a,b+1}-C_{a+1,b}+C_{a+1,b+1}.
\label{eqS:D-map}
\end{align}
Every additive one-label array lies in $\ker D$. Conversely, if $DC=0$, set
\begin{equation}
u_a=C_{a1},\qquad v_b=C_{1b}-C_{11}.
\end{equation}
Vanishing mixed differences imply recursively
\begin{equation}
C_{ab}=C_{a1}+C_{1b}-C_{11}=u_a+v_b,
\end{equation}
so $\ker D=\Lambda_{\rm one}$.

The map $D$ is surjective over the integers. Given $x_{ab}\in\mathbb Z$, impose $C_{a1}=C_{1b}=0$ and define
\begin{equation}
C_{ab}=\sum_{i<a}\sum_{j<b}x_{ij},\qquad a,b\ge2.
\label{eqS:discrete-integral}
\end{equation}
Then $DC=x$. Hence
\begin{equation}
\frac{\Lambda}{\Lambda_{\rm one}}\simeq\mathbb Z^{(N_A-1)(N_B-1)}.
\label{eqS:quotient-proof}
\end{equation}
The rank-one quotient is torsion-free. Other rectangular-difference bases are related by unimodular transformations.

\begin{proposition}[Universal property of the quotient]
Let $\mathcal G$ be an Abelian group and let $f:\Lambda\to\mathcal G$ be an additive resolved invariant. If $f\circ\Phi=0$, equivalently if $f$ vanishes on every $A$-only and $B$-only table, then there exists a unique homomorphism $\bar f:\mathcal Q_{AB}\to\mathcal G$ such that
\begin{equation}
f=\bar f\circ\pi.
\label{eqS:universal-factorization}
\end{equation}
\end{proposition}
\begin{proof}
Let $\pi:\Lambda\to\Lambda/\Lambda_{\rm one}=\mathcal Q_{AB}$ be the quotient map. Because $f$ vanishes on $\Lambda_{\rm one}=\operatorname{im}\Phi$, the value $\bar f(\pi(C))=f(C)$ is independent of the representative. Surjectivity of $\pi$ makes the factorization unique.
\end{proof}
For $N_A=N_B=2$, the unique primitive coordinate, up to sign, is $C_{++}-C_{+-}-C_{-+}+C_{--}$.

The quotient has a determinant-line realization.

For the rank-one decomposition $E=\bigoplus_{ab}L_{ab}$, define
\begin{equation}
\mathcal L_{ab}=L_{ab}\otimes L_{a,b+1}^{-1}\otimes L_{a+1,b}^{-1}\otimes L_{a+1,b+1}.
\label{eqS:mixed-line}
\end{equation}
Tensoring every sector with $A_a\otimes B_b$ adds an element of $\Lambda_{\rm one}$ and leaves $c_1(\mathcal L_{ab})$ invariant. Equation~\eqref{eqS:quotient-proof} therefore realizes the mixed quotient at the determinant-line level.

\subsection{Label resolution and entanglement spectrum}\label{app:label-resolution}
\label{secS:compressed}

For a binary microscopic label with $Q_+=(1+Z)/2$,
\begin{equation}
X=PZP=2PQ_+P-P.
\label{eqS:compressed-binary}
\end{equation}
If $\xi_j$ are the nonzero eigenvalues of $Q_+PQ_+$ and $\mu_j$ the corresponding eigenvalues of $PZP$ in $\operatorname{Ran}P$, then
\begin{equation}
\mu_j=2\xi_j-1,
\qquad
\varepsilon_j^{\rm ent}=\log\frac{1-\xi_j}{\xi_j}=-2\operatorname{artanh}\mu_j.
\label{eqS:entanglement-relation}
\end{equation}
At each binary stage, a compressed-label gap closes exactly when the corresponding internal-partition spectrum reaches $\xi=1/2$. This gives the binary form of the projected-feature diagnostic \cite{Wang2023Feature,Hung2026Nested}.

A label-resolved decomposition can be constructed hierarchically. First split $P$ with $X_A=PZ_AP$ into rank-two projectors $P_a$, then diagonalize
\begin{equation}
X_{B|a}=P_aZ_BP_a.
\label{eqS:conditional-label-operator}
\end{equation}
The gaps $\delta_A$ and $\delta_{B|a}$ identify which stage of the decomposition fails.

For the mixed Hamiltonian at $\epsilon=6/5$, we also scan
\begin{equation}
X_{\rm lab}(\alpha,\beta)=P(\alpha Z_A+\beta Z_B)P.
\label{eqS:label-operator-family}
\end{equation}
The heat map in Fig.~\ref{figS:label-resolution}(a), evaluated on a $301\times301$ $(\alpha,\beta)$ grid with a $25\times25$ Brillouin-zone grid, uses the normalized direction $(\alpha,\beta)/\sqrt{\alpha^2+\beta^2}$ and records the minimum eigenvalue separation over the Brillouin zone. The region $\alpha>\beta>0$ contains the point $(\alpha,\beta)=(2,1)$. Along $\alpha=1$ and $0.15\le\beta\le0.85$,
\begin{equation}
0.15\lesssim\delta_{\rm lab}\lesssim0.69,
\label{eqS:label-resolution-path}
\end{equation}
while the ordered FHS table is
\begin{equation}
(C_{++},C_{+-},C_{-+},C_{--})=(1,-1,-1,1),
\qquad \chi=4.
\label{eqS:label-resolution-table}
\end{equation}
At $\epsilon=6/5$, the hierarchical compressed gaps are
\begin{equation}
\delta_A\simeq1.4,
\qquad \delta_{B|A=\pm}\simeq1.6.
\label{eqS:hierarchical-gaps}
\end{equation}
The corresponding minimum distances of the binary entanglement eigenvalues from $1/2$ are approximately $0.35$ and $0.39$. At each binary stage the exact relation is $\xi=(1+\mu)/2$. The scalar four-line resolution and the hierarchical binary resolution answer different questions: a collision between two scalar spectral islands can involve more than one binary step.

\begin{figure*}[t]
\centering
\includegraphics[width=0.99\textwidth]{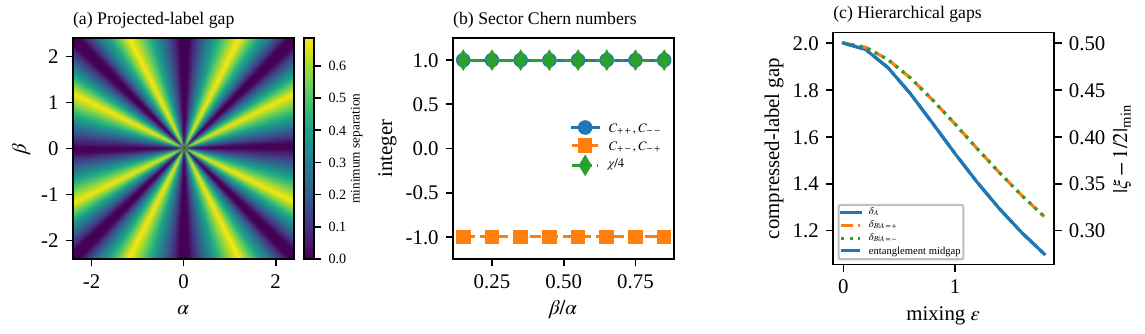}
\caption{\textbf{Projected-label gaps and hierarchical diagnostics.}
(a) Minimum normalized label-resolution gap of $P(\alpha Z_A+\beta Z_B)P$ at $\epsilon=6/5$, computed on a $301\times301$ parameter grid and a $25\times25$ Brillouin-zone grid. (b) Along $\alpha=1$ and $0.15\le\beta\le0.85$, the four sector Chern integers and $\chi$ remain constant. (c) The sequential $A$ and conditional-$B$ compressed gaps are open along the mixing path; the secondary axis shows the associated entanglement-spectrum distance from $\xi=1/2$.}
\label{figS:label-resolution}
\end{figure*}

\subsection{Gauss--Codazzi formula and numerical convergence}\label{app:gauss-codazzi}
\label{secS:sector-curvature}

In Eq.~\eqref{eq:mixing-path}, the entries of $\bm v_\star$ correspond to the coefficients in Eq.~\eqref{eq:label-mixing} in the order
\begin{equation*}
(v_A,v_B,v_{AB}^{(\sigma)},v_{AB}^{(0)}).
\end{equation*}
They give comparable weight to the two single-label rotations and smaller amplitudes to the joint channels. The parameter $\epsilon$ controls the common scale. The energy and label-resolution gaps remain open along this path.

Let $P$ be the occupied projector, let $Q_s\le P$ be a smooth resolved subprojector, and write $L_s=\operatorname{Ran}Q_s$. In a local occupied frame, write the Wilczek--Zee connection and curvature \cite{WilczekZee1984} as $\cA$ and $\cF=\dd\cA+\cA\wedge\cA$. The covariant derivative is
\begin{equation}
D_{\cA}Q_s=\dd Q_s+[\cA,Q_s].
\end{equation}
Choose a frame adapted locally to $Q_s\oplus(P-Q_s)$. Including the factor $1/(2\pi\ii)$, so that the integral is a Chern number, the covariant Gauss--Codazzi identity is
\begin{align}
f_s&=\frac{1}{2\pi\ii}
\left\{\Tr(Q_s\cF)+\Tr\left[Q_s(D_{\cA}Q_s)^{\wedge2}\right]\right\},\nonumber\\
\int_X f_s&=c_1(L_s)[X].
\label{eqS:gauss-codazzi}
\end{align}
In the ambient Hilbert space,
\begin{equation}
\cF\leftrightarrow P(\dd P)^{\wedge2}P,
\qquad
D_{\cA}Q_s\leftrightarrow P(\dd Q_s)P.
\label{eqS:ambient-identities}
\end{equation}
For integer weights $w_s$, define $H_w=\sum_sw_sQ_s$. Then
\begin{align}
f_w&=\sum_sw_sf_s=f_{\rm proj}+\Xi_w,\nonumber\\
f_{\rm proj}&=\frac{1}{2\pi\ii}\Tr(H_w\cF),\nonumber\\
\Xi_w&=\frac{1}{2\pi\ii}\sum_sw_s\Tr\left[Q_s(D_{\cA}Q_s)^{\wedge2}\right].
\label{eqS:weighted-curvature}
\end{align}
For the binary weights used in the main text, Eq.~\eqref{eqS:weighted-curvature} reduces to Eq.~\eqref{eq:curvature-decomposition}.

On periodic $N\times N$ momentum grids, we diagonalize the full $8\times8$ Hamiltonian, form the rank-four occupied projector, and diagonalize $P(2Z_A+Z_B)P$ within the occupied space. Its nondegenerate eigenvectors define four rank-one projectors in the order $(++,+-,-+,--)$. Sector Chern integers are computed independently with the Fukui--Hatsugai--Suzuki plaquette formula \cite{FukuiHatsugaiSuzuki2005}. Projector derivatives in Eq.~\eqref{eqS:ambient-identities} are evaluated by Fourier spectral differentiation.

Here $\mathcal R(\bm k)=f_{AB}-f_{\rm proj}-\Xi_{AB}$. At $\epsilon=0$, $\Xi_{AB}$ vanishes to numerical precision. Along the mixing path it becomes finite and completes the determinant-line curvature. Table~\ref{tabS:curvature-convergence} summarizes the endpoint $\epsilon=6/5$.
\begin{table}[ht]
\caption{Spectral-grid convergence at $\epsilon=6/5$. The projected and Gauss--Codazzi contributions sum to the Chern number $4$.}
\label{tabS:curvature-convergence}
\begin{ruledtabular}
\begin{tabular}{rccc}
$N$ & $\int f_{\rm proj}$ & $\int\Xi_{AB}$ & $\|\mathcal R\|_\infty$\\
\hline
15 & 3.83 & 0.17 & $5\times10^{-3}$\\
51 & 3.84 & 0.16 & $7\times10^{-6}$\\
151 & 3.84 & 0.16 & $1\times10^{-13}$
\end{tabular}
\end{ruledtabular}
\end{table}

\section{Boundary and finite-time pumping}\label{app:pump}
\label{secS:edge-pump}

\subsection{Boundary anomaly and ribbon model}

The edge data are
\begin{align}
K&=\operatorname{diag}(1,-1,-1,1),\nonumber\\
t_A&=(1,1,-1,-1)^T,
& t_B&=(1,-1,1,-1)^T.
\end{align}
The mixed anomaly coefficient is $t_A^TK^{-1}t_B=4$. A complete set of symmetric Haldane null vectors would have to satisfy
\begin{equation}
\ell_i^TK^{-1}\ell_j=0,
\qquad
\ell_i^TK^{-1}t_A=\ell_i^TK^{-1}t_B=0.
\label{eqS:null-conditions}
\end{equation}
Such a complete isotropic set would force the anomaly pairing of $t_A$ and $t_B$ to vanish. Its nonzero value therefore obstructs a unique, symmetric, short-range-entangled edge. Breaking either label removes one neutrality constraint and permits a mass.

The real-space QWZ ribbon is
\begin{align}
H_M(k_y)=&\sum_{x=0}^{L_x-1}c_x^\dagger
\left[\sin k_y\,\sigma_y+(M+\cos k_y)\sigma_z+v_xI_2\right]c_x\nonumber\\
&+\sum_{x=0}^{L_x-2}\left[c_x^\dagger T_xc_{x+1}+\mathrm{H.c.}\right],
\qquad T_x=\frac{\sigma_z-\ii\sigma_x}{2}.
\label{eqS:full-ribbon}
\end{align}
For Fig.~\ref{fig:edge-pump}(a,b), we diagonalize the $8L_x\times8L_x$ ribbon Hamiltonian at $L_x=14$. A weak scalar disorder profile $v_x$, label-preserving edge potentials, and sector-dependent momentum offsets separate accidental crossings without breaking $U(1)_A\times U(1)_B$. We use
\begin{equation}
\delta_{ab}=\begin{cases}0,&ab=+1,\\ \pi,&ab=-1,\end{cases}
\label{eqS:sector-offsets}
\end{equation}
which is a gauge-equivalent relabeling of transverse momenta for even circumference and is introduced only to separate coincident edge crossings. The plotted states are selected from the full ribbon spectrum by their weight in the outer three layers. Adding a boundary $0.25\tau_x$ term gives a minimum absolute edge energy of approximately $0.15$, corresponding to an edge gap of $0.30$.

\subsection{Pump protocol and robustness}

The thermodynamic spectral-flow calculation on the open-$x$, periodic-$y$ cylinder transfers four units of $B$ charge while leaving the total and $A$ charges unchanged, in agreement with the finite-time protocol below. The boundary twist and momentum offsets are defined in Eq.~\eqref{eqS:sector-offsets} and the preceding discussion; the transported $B$ density is localized on the two open boundaries.

To connect the spectral-flow branch to a realizable protocol, we evolve every initially occupied one-particle state under
\begin{equation}
\ii\partial_t|\psi_j(t)\rangle=H[\theta_A(t)]|\psi_j(t)\rangle,
\qquad
\theta_A(t)=\theta_0+2\pi t/T.
\label{eqS:finite-time-schrodinger}
\end{equation}
The calculation shown in Fig.~\ref{fig:edge-pump} uses $\theta_0=0.37$, $L_x=10$, $L_y=4$, $720$ midpoint time steps, weak $x$-dependent scalar disorder, and all initially occupied one-particle states in each conserved-label sector. The offset places the initial Hamiltonian away from the small finite-size edge anticrossing. Each midpoint propagator is evaluated by exact matrix exponentiation. The evolved occupied orbitals define
\begin{equation}
P_{\rm evo}(t)=\sum_{j\in\mathrm{occ}}|\psi_j(t)\rangle\langle\psi_j(t)|.
\label{eqS:evolved-projector}
\end{equation}
The right-half observables are
\begin{align}
Q_{\rm tot}^R(t)&=\Tr[P_{\rm evo}(t)\Pi_R],\nonumber\\
Q_A^R(t)&=\Tr[P_{\rm evo}(t)\Pi_R Z_A],\nonumber\\
Q_B^R(t)&=\Tr[P_{\rm evo}(t)\Pi_R Z_B].
\label{eqS:finite-time-observable}
\end{align}
For $L_x=10$, $L_y=4$, $T=100$, and $720$ midpoint steps, the total and $A$ transfers are below $10^{-3}$, while the $B$ transfer is $4$.

Doubling the temporal resolution leaves $\Delta Q_B^R$ unchanged at the quoted precision. The shaded region in Fig.~\ref{fig:edge-pump}(d) contains the sampled ramp times between $20$ and $500$ for which $|\Delta Q_B^R-4|<0.02$.

The sampled offsets in Fig.~\ref{figS:pump-robustness}(d) form a plateau around the chosen value. Starting close to the finite-size crossing, as at $\theta_0=0.15$, moves the evolution outside the bulk-adiabatic, edge-diabatic window for this ramp time.

For this finite cylinder, the minimum bulk-like gap is
\begin{equation}
\Delta_{\rm bulk}\simeq2,
\end{equation}
whereas the minimum opposite-edge anticrossing is
\begin{equation}
\Delta_{\rm edge}\simeq2\times10^{-10}.
\end{equation}
These scales leave a broad ramp-time window that follows the bulk gap and crosses the opposite-edge anticrossing. Among the sampled ramp times from $20$ to $500$, the maximum observed deviation of $\Delta Q_B^R$ from four is $0.02$. The transferred density integrates to opposite edge-localized values close to $-4$ and $+4$.

The calculation with $v_x$ disorder exploits transverse translation. As an independent check, we also evolve the full real-space Hamiltonian with two-dimensional label-preserving disorder $v_{x,y}$ on an $8\times4$ cylinder. At disorder amplitude $0.02$ and $T=100$, the unwanted total and $A$ transfers are below $10^{-3}$, while the $B$ transfer is $4$. The full real-space calculation gives the same transfer without transverse-momentum block diagonalization.

At $T=100$, the four cylinders $8\times4$, $10\times4$, $12\times6$, and $14\times6$ all give $\Delta Q_B^R=4$ to the displayed precision. For twelve independent realizations of two-dimensional label-preserving disorder $v_{x,y}$ with amplitude $0.02$ on the $8\times4$ cylinder, the mean $B$ transfer is $4$, with a standard deviation below $10^{-3}$.
The opposite-edge hybridization of a clean Hamiltonian in the same phase obeys
\begin{equation}
\Delta_{\rm edge}(L_x)\propto e^{-L_x/\xi},
\qquad
\xi\simeq1.4,
\label{eqS:edge-exponential}
\end{equation}
so the bulk-adiabatic, edge-diabatic ramp-time window broadens with width.

\begin{figure*}[t]
\centering
\includegraphics[width=0.98\textwidth]{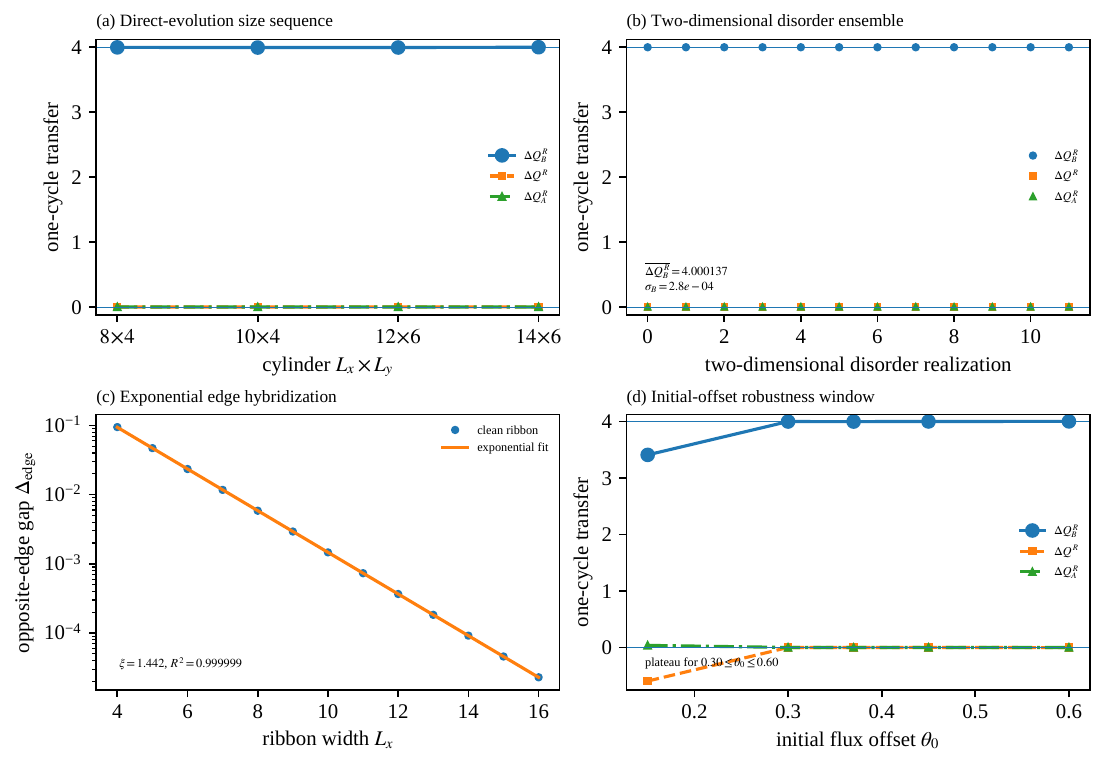}
\caption{\textbf{Robustness of crossed transport.}
(a) The finite-time $B$ transfer is near four across four cylinder sizes. (b) Twelve $v_{x,y}$ disorder realizations cluster tightly around four; the line and band show the mean and one standard deviation. (c) Opposite-edge hybridization decays exponentially with width in a clean Hamiltonian from the same phase, explaining why the bulk-adiabatic, edge-diabatic ramp window broadens with $L_x$. (d) Initial-flux-offset scan at fixed ramp time. The sampled plateau contains the chosen $\theta_0=0.37$ value; the point $\theta_0=0.15$ lies too close to the finite-size crossing for the chosen ramp time.}
\label{figS:pump-robustness}
\end{figure*}

\section{Proof of the factorization criterion}
\label{app:factorization-proof}

\begin{proof}[Proof of Theorem~\ref{thm:factorization}]
Cover $S^4$ by northern and southern four-balls. Since $H^2(S^4;\mathbb Z)=0$, the determinant line of $E$ is topologically trivial. After choosing a determinant trivialization, $E$ is specified by a clutching map
\begin{equation}
\begin{aligned}
g&:S^3\longrightarrow SU(pq),\\
W_3[g]&=\frac{1}{24\pi^2}\int_{S^3}\Tr(g^{-1}\dd g)^{\wedge3}=C_2(E).
\end{aligned}
\label{eq:appendix-clutching}
\end{equation}
Here $W_3$ identifies $\pi_3[SU(pq)]\simeq\mathbb Z$ with the second Chern number.

Suppose first that $E\simeq E_A\otimes E_B$. Because $H^2(S^4;\mathbb Z)=0$, the determinant lines of $E_A$ and $E_B$ are also trivial, and their clutching maps may be chosen as $g_A:S^3\to SU(p)$ and $g_B:S^3\to SU(q)$. Writing $\omega_A=g_A^{-1}\dd g_A$ and $\omega_B=g_B^{-1}\dd g_B$ gives
\begin{equation}
(g_A\otimes g_B)^{-1}\dd(g_A\otimes g_B)=\omega_A\otimes I_q+I_p\otimes\omega_B.
\end{equation}
Since $\Tr\omega_A=\Tr\omega_B=0$, the mixed terms in the cubic trace vanish, and Eq.~\eqref{eq:tensor-winding-main} follows. Hence every factorized bundle has
\begin{equation}
C_2(E)\in q\mathbb Z+p\mathbb Z=\gcd(p,q)\mathbb Z.
\label{eq:appendix-necessary}
\end{equation}
The tensor-product transition functions lie in $G_{p,q}$, so (ii) implies (i) and (iii).

Conversely, assume $C_2(E)=\gcd(p,q)m$. Choose B\'{e}zout integers $r,s$ such that $rp+sq=\gcd(p,q)$. Since $\pi_3[SU(n)]\simeq\mathbb Z$ for $n\ge2$, there exist rank-$p$ and rank-$q$ determinant-trivial bundles with clutching windings $sm$ and $rm$. Their tensor product has winding
\begin{equation}
q(sm)+p(rm)=\gcd(p,q)m=C_2(E).
\end{equation}
Complex $SU(pq)$ bundles over $S^4$ are classified by their clutching integer, so this tensor-product bundle is isomorphic to $E$. Thus (iii) implies (ii).

We next relate a reduction to $G_{p,q}$ to an honest tensor factorization. The kernel of
\begin{equation}
SU(p)\times SU(q)\longrightarrow G_{p,q}
\end{equation}
is the finite central group $\mu_d$, with $d=\gcd(p,q)$. The obstruction to lifting a $G_{p,q}$ reduction to an $SU(p)\times SU(q)$ bundle lies in $H^2(S^4;\mathbb Z_d)$, which vanishes. Therefore every reduction lifts and produces bundles $E_A,E_B$ with $E\simeq E_A\otimes E_B$, proving (i)$\Rightarrow$(ii).

For $p=q$, allowing the two factors to be exchanged adds a discrete $\mathbb Z_2$ choice. Its obstruction lies in $H^1(S^4;\mathbb Z_2)=0$, so ordered and unordered factorizations give the same criterion on $S^4$. Finally, for $p=q=2$, let $E_{\mathbb R}$ denote the underlying real bundle. Its fourth Stiefel--Whitney class is the mod-two reduction of $c_2(E)$,
\begin{equation}
c_2(E)\bmod2=w_4(E_{\mathbb R}).
\end{equation}
Evaluating this class on $[S^4]$ gives $\eta_{2,2}(E)$ \cite{MilnorStasheff1974}.
\end{proof}

\clearpage

\section{Eight-level realization, three-transmon compilation, and tomography}\label{app:eight-level}
\label{secS:full-realization}

\subsection{Eight-level model and second Chern number}

The surface and clutching calculations use eight-state Hilbert spaces with rank-four occupied multiplets. The two-dimensional basis is $\mathbb C_A^2\otimes\mathbb C_B^2\otimes\mathbb C_{\rm orb}^2$. For the four-dimensional model we choose a concrete superconducting encoding,
\begin{equation}
\mathcal H_{\rm amb}=\mathbb C_\kappa^2\otimes\mathbb C_\tau^2\otimes\mathbb C_\sigma^2
\simeq\mathcal H_{q_\kappa}\otimes\mathcal H_{q_\tau}\otimes\mathcal H_{q_\sigma},
\label{eqS:ambient-hardware}
\end{equation}
with computational states $|\kappa\tau\sigma\rangle\leftrightarrow|q_\kappa q_\tau q_\sigma\rangle$. A linear device with $q_\tau$ between $q_\kappa$ and $q_\sigma$ supplies all required two-qubit links.

A convenient Clifford representation in the $\tau\otimes\sigma$ sector is
\begin{equation}
\begin{aligned}
\Gamma_1&=\tau_x\sigma_x, &
\Gamma_2&=\tau_x\sigma_y, &
\Gamma_3&=\tau_x\sigma_z,\\
\Gamma_4&=\tau_y, &
\Gamma_5&=\tau_z.&&
\end{aligned}
\label{eqS:gamma-representation}
\end{equation}
They satisfy $\{\Gamma_\mu,\Gamma_\nu\}=2\delta_{\mu\nu}$. In the fixed ambient basis, the flattened odd Hamiltonian has the compact Pauli-string form
\begin{equation}
\begin{aligned}
H_{\rm odd}^{\rm flat}(\bm n)
={}&\frac{I_2+\kappa_z}{2}\Big(
 n_1\tau_x\sigma_x+n_2\tau_x\sigma_y\\
&\hspace{2.2cm}+n_3\tau_x\sigma_z+n_4\tau_y+n_5\tau_z\Big)\\
&-\frac{I_2-\kappa_z}{2}\tau_z.
\end{aligned}
\label{eqS:odd-flat-explicit}
\end{equation}
The first $\kappa_z=+1$ block is the flattened Yang Hamiltonian. In the $\kappa_z=-1$ block, $-\tau_z$ supplies two occupied and two unoccupied spectator states. Equation~\eqref{eqS:odd-flat-explicit} expands into eight off-diagonal pair couplings and two diagonal detunings.

In the ambient basis of Eq.~\eqref{eqS:ambient-hardware}, the interblock-coupled Hamiltonian is
\begin{equation}
H_\lambda(\bm n)=H_{\rm odd}^{\rm flat}(\bm n)+\lambda V_{\rm cpl}(\bm n),
\label{eqS:full-structured-H}
\end{equation}
with
\begin{equation}
V_{\rm cpl}(\bm n)=\frac{1}{2}\left(\frac{7}{20}M_0+\frac{1}{\sqrt5}\sum_{\mu=1}^{5}n_\mu M_\mu\right),
\label{eqS:structured-s4-coupling}
\end{equation}
where
\begin{align}
(M_0,\ldots,M_5)=(&\kappa_x,\kappa_y\tau_x,\kappa_y\tau_z\sigma_x,\nonumber\\
&\kappa_x\tau_y\sigma_z,\kappa_z\tau_x\sigma_y,\tau_y\sigma_x).
\label{eqS:structured-pauli-strings}
\end{align}
Every matrix is traceless, Hermitian, and has operator norm one. Equations~\eqref{eqS:odd-flat-explicit} and \eqref{eqS:structured-pauli-strings} define the finite-level Hamiltonian used below. In the ordered basis $|\kappa\tau\sigma\rangle$, the six $M_\mu$ define conditional pair couplings, while Eq.~\eqref{eqS:odd-flat-explicit} supplies the Yang-block couplings and two diagonal detunings. 
By the triangle and Cauchy--Schwarz inequalities,
\begin{align}
\|V_{\rm cpl}(\bm n)\|&\le\frac12\left(\frac{7}{20}+\frac{1}{\sqrt5}\sum_\mu|n_\mu|\right)\nonumber\\
&\le\frac{27}{40}\simeq0.68.
\label{eqS:operator-bound}
\end{align}
The unperturbed flattened occupied and unoccupied eigenvalues are $-1$ and $+1$. Weyl's inequality therefore gives, uniformly for every $\bm n\in S^4$,
\begin{equation}
\Delta_\lambda(\bm n)\ge2-\frac{27}{20}|\lambda|.
\label{eqS:global-gap-bound}
\end{equation}
At $\lambda=5/4$, the analytic lower bound is $5/16$, so the interpolation is gapped and homotopy invariance fixes $C_2=1$. A direct sphere sample gives a minimum gap of approximately $1.2$, and numerical integration verifies the second Chern number.

\subsection{Native-gate compilation on a three-transmon chain}

Every term in Eqs.~\eqref{eqS:odd-flat-explicit} and \eqref{eqS:structured-pauli-strings} is a Pauli word of weight at most three. Let $\mathsf P=\bigotimes_jP_j$ be a Pauli word with nontrivial support $S$ and weight $w=|S|$. Choose single-qubit basis changes $B_X=H$, $B_Y=HS^\dagger$, and $B_Z=I$, and set $B(\mathsf P)=\bigotimes_j B_{P_j}$ on the support, so that $B(\mathsf P)\mathsf P B(\mathsf P)^\dagger=\prod_{j\in S}Z_j$. If $t\in S$ is a target and
\begin{equation}
C_S=\prod_{j\in S\setminus\{t\}}\mathrm{CNOT}_{j\rightarrow t},
\end{equation}
then the exact Pauli rotation is
\begin{equation}
\exp\left(-\frac{\ii\theta}{2}\mathsf P\right)
=B(\mathsf P)^\dagger C_S^\dagger R_z^{(t)}(\theta)C_SB(\mathsf P).
\label{eqS:pauli-rotation-compilation}
\end{equation}
It uses $2(w-1)$ CNOT-equivalent entanglers and only local gates otherwise. Figure~\ref{fig:transmon-circuits}(a) gives the explicit circuit for $\mathsf P=M_2=YZX$: the state first encounters $S^\dagger$ followed by $H$ on $q_\kappa$ and $H$ on $q_\sigma$, the two adjacent CNOTs accumulate the $ZZZ$ parity on $q_\tau$, and the sequence is reversed after $R_z(\theta)$. Fixed-frequency transmon processors provide calibrated all-microwave entanglers that are locally equivalent to CNOT, as well as more general native two-qubit gates \cite{Chow2011,Wei2024}.

\begin{center}
\refstepcounter{table}\label{tabS:transmon-compilation}
\parbox{\columnwidth}{\small\textbf{TABLE~\thetable.} Compilation of the six interblock controls on the ordered chain $(q_\kappa,q_\tau,q_\sigma)$. The last column counts CNOT-equivalent two-qubit entanglers for one exact Pauli rotation.}\par\vspace{3pt}
\begin{ruledtabular}
\begin{tabular}{cccc}
control & Pauli word & weight & $N_{2q}$\\
\hline
$M_0$ & $XII$ & $1$ & $0$\\
$M_1$ & $YXI$ & $2$ & $2$\\
$M_2$ & $YZX$ & $3$ & $4$\\
$M_3$ & $XYZ$ & $3$ & $4$\\
$M_4$ & $ZXY$ & $3$ & $4$\\
$M_5$ & $IYX$ & $2$ & $2$
\end{tabular}
\end{ruledtabular}
\end{center}

For the three weight-three controls, choosing $q_\tau$ as the parity target uses the two adjacent device edges and requires no SWAP gates. The two weight-two terms lie on those edges as well. The ten Pauli words in $H_{\rm odd}^{\rm flat}$ have weights one, two, or three and are compiled by the same construction. Thus the digital realization needs no native three-body interaction. A product formula can synthesize the Hamiltonian path. For tomography, each occupied eigenvector can instead be compiled directly into a three-qubit preparation circuit, avoiding a sequence of Hamiltonian rotations.

Parameterize $S^4$ by $\bm y=(\alpha,\beta,\gamma,\phi)$:
\begin{equation}
\begin{aligned}
n_1&=\sin\alpha\sin\beta\sin\gamma\cos\phi,\qquad
n_2=\sin\alpha\sin\beta\sin\gamma\sin\phi,\\
n_3&=\sin\alpha\sin\beta\cos\gamma,\qquad
n_4=\sin\alpha\cos\beta,\\
n_5&=\cos\alpha.
\end{aligned}
\label{eqS:s4-coordinates}
\end{equation}
For an isolated occupied multiplet with eigenstates $|i\rangle$ and unoccupied states $|a\rangle$, projector derivatives are evaluated without choosing a gauge:
\begin{equation}
\partial_\mu P=
\sum_{i\in\mathrm{occ}}\sum_{a\in\mathrm{unocc}}
\frac{|a\rangle\langle a|\partial_\mu H|i\rangle\langle i|}{E_i-E_a}
+\mathrm{H.c.}
\label{eqS:projector-derivative}
\end{equation}
Set $F_{\mu\nu}=P[\partial_\mu P,\partial_\nu P]P$. With the chosen orientation,
\begin{equation}
C_2=-\frac{1}{4\pi^2}\int\dd^4y\,
\Tr\!\left(F_{12}F_{34}-F_{13}F_{24}+F_{14}F_{23}\right).
\label{eqS:direct-c2}
\end{equation}
Gauss--Legendre quadrature is used for $\alpha,\beta,\gamma$ and a uniform midpoint grid for $\phi$. At $\lambda=5/4$, the eighth-order quadrature gives $C_2=1$; convergence is shown in Fig.~\ref{fig:full-reduction}(a). The analytic bound keeps the gap open along the path, and the block commutator grows with $\lambda$, showing that exact block conservation is broken.

The additive perturbation generally splits the four occupied energies. Numerical $C_2$ evaluation and occupied-subspace tomography use only the isolated rank-four projector $P_\lambda$. Reflattening gives an exactly degenerate Hamiltonian with the same occupied projector when a Wilczek--Zee or holonomic interpretation is desired:
\begin{equation}
\widetilde H_\lambda(\bm n)=I_8-2P_\lambda(\bm n).
\label{eqS:reflattened-hamiltonian}
\end{equation}
It preserves the interblock-coupled projector, $C_2$, the transition function, and the factorization obstruction while restoring occupied energy $-1$ and unoccupied energy $+1$.

\subsection{Clutching map and winding reconstruction}\label{app:tomography}

For the standard north--south gauges of the Yang bundle,
\begin{equation}
u(\bm x)=x_4I_2+\ii\sum_{j=1}^{3}x_j\sigma_j,
\qquad \bm x\in S^3.
\label{eqS:yang-clutching}
\end{equation}
The odd map $u\oplus I_2$ has winding one, whereas the pointwise factorized map $u\otimes I_2$ has winding two. Their residues modulo two distinguish the odd bundle from the factorized reference.

The analytic maps $u\oplus I_2$ and $u\otimes I_2$ provide reference clutching maps, while the reconstruction shown in Fig.~\ref{fig:full-reduction}(c) is performed on the Pauli-string-coupled occupied projector. Let $R_N$ and $R_S$ be fixed four-column reference frames at the north and south poles. Wherever $R_\nu^\dagger P R_\nu$ is positive, polar projection gives an orthonormal occupied frame
\begin{equation}
U_\nu=P R_\nu(R_\nu^\dagger P R_\nu)^{-1/2}.
\label{eqS:polar-patch-frame}
\end{equation}
On the equator the transition matrix
\begin{equation}
G_{NS}=U_N^\dagger U_S\in U(4)
\label{eqS:tomographic-transition}
\end{equation}
is reconstructed from the two patch frames. The $U(4)$-valued transition matrix is kept without determinant normalization. Because $\pi_3[U(1)]=0$, the determinant component does not contribute, and the standard trace integral computes the $\pi_3[U(4)]\simeq\mathbb Z$ winding directly.

\begin{proposition}[Gauge invariance of the clutching winding]
Let $U_N,U_S$ be regular frames on the two closed hemispheres and let $G_{NS}=U_N^\dagger U_S$ on their common $S^3$ boundary. Under arbitrary regular patch-frame changes $U_N\mapsto U_Nh_N$ and $U_S\mapsto U_Sh_S$, the winding $W_3[G_{NS}]$ is unchanged.
\end{proposition}
\begin{proof}
The transition map transforms as $G_{NS}\mapsto h_N^\dagger G_{NS}h_S$. On a closed three-manifold, the Polyakov--Wiegmann identity gives
\begin{equation}
W_3[h_N^\dagger G_{NS}h_S]
=W_3[G_{NS}]-W_3[h_N]+W_3[h_S],
\label{eqS:PW-patch}
\end{equation}
because the remaining cross term is exact. Each $h_\nu$ extends over its hemisphere $B^4$, so its restriction to the boundary is null-homotopic and $W_3[h_\nu]=0$. Hence $W_3$ is independent of the reference frames, polar construction, and smooth north--south frame changes. The integer $W_3$ is a property of the occupied clutching class; the candidate factorization enters when it is reduced modulo the image $q\mathbb Z+p\mathbb Z$.
\end{proof}

Clean and noisy reconstructions use a common discretization. Parametrize the equatorial $S^3$ by two polar angles $\vartheta,\theta\in[0,\pi]$ and one periodic azimuth $\phi\in[0,2\pi)$. For order $N$, use $N$ mapped Gauss--Legendre nodes in each polar coordinate and $2N$ equispaced azimuthal nodes. The transition matrix is reconstructed once at every point, producing a single array with
\begin{equation}
\begin{aligned}
N_{\rm eq}=N_P=N_G&=2N^3,\\
N_{\rm pole}&=
\begin{cases}
0, & \text{pre-calibrated references},\\
2, & \text{measured once}.
\end{cases}
\end{aligned}
\label{eqS:sample-accounting}
\end{equation}
Here $N_P$ and $N_G$ are the numbers of equatorial projector reconstructions and transition matrices. Optional pole references add two one-time projector reconstructions. All derivatives are formed from the same equatorial grid: polar derivatives use global barycentric differentiation on the Gauss--Legendre nodes, and the periodic derivative uses the discrete Fourier transform. With $A_i=G_{NS}^\dagger\partial_iG_{NS}$, the quadrature is
\begin{equation}
\widehat W_3^{\rm grid}
=-\frac{1}{8\pi^2}
\sum_{ijk}w_i^{\vartheta}w_j^{\theta}w^{\phi}
\operatorname{Re}\Tr\!\left[A_\vartheta[A_\theta,A_\phi]\right]_{ijk}.
\label{eqS:grid-winding}
\end{equation}
The formula is applied without modification to the clean and finite-shot reconstructions.

At $N=6$, $432$ projector reconstructions per bundle give windings that converge to $1$ and $2$ for the odd bundle and factorized reference. A Sobol scan of approximately $2\times10^5$ points gives positive overlap eigenvalues throughout the sampled hemispheres, with minima $0.36$ and $0.32$. This provides a numerical conditioning check.

\begin{figure*}[t]
\centering
\includegraphics[width=0.82\textwidth]{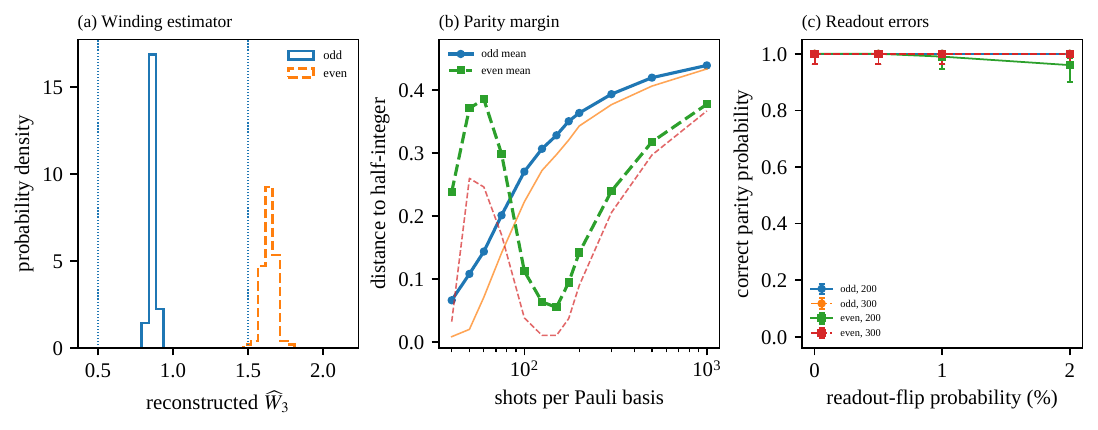}
\caption{\textbf{Three-transmon finite-shot tomography and readout errors.}
(a) Reconstructed winding distributions at $200$ pooled shots per Pauli basis. For the $N=4$ grid this is approximately $7\times10^5$ shots per bundle. The dashed half-integers are the calibrated parity boundaries; the shift from the clean integers shows the coarse-grid estimator bias. (b) Distance to the nearest half-integer decision boundary versus shot budget. (c) Correct parity under independent symmetric readout flips on the three transmon readout bits. At $300$ shots per basis, both the odd bundle and the factorized reference are classified correctly in all $100$ trials through $2\%$ confusion.}
\label{figS:tomography-systematics}
\end{figure*}

\subsection{Finite-shot tomography}

At each parameter point, tomography reconstructs the rank-four occupied projector $P$. Spectral flattening makes the occupied quartet exactly degenerate. For a three-qubit proof of principle, the $8\times8$ Hamiltonian is diagonalized classically and the four occupied eigenvectors are compiled offline into exact state-preparation circuits \cite{Mottonen2005}. Pooling equal numbers of shots prepares the ensemble
\begin{equation}
\rho_{\rm occ}=\frac{P}{4}=\frac14\sum_{j=1}^{4}|u_j\rangle\langle u_j|.
\label{eqS:rho-occ}
\end{equation}
With unrestricted pairwise CNOT connectivity, any pure three-qubit state can be prepared from $|000\rangle$ using at most three CNOTs; on the linear chain used here, a fourth may be required in the worst case \cite{Znidaric2008}. The worst-case CNOT count for direct preparation is therefore no larger than that of one weight-three Pauli rotation, and it avoids a product-formula sequence of Hamiltonian rotations.

A device can reproduce this ensemble by selecting $j$ uniformly for each shot. A more convenient implementation assigns $n_s/4$ shots to each of the four preparation circuits and pools the counts. To compare the two procedures, let $\bm p_j$ be the eight-outcome distribution of $|u_j\rangle$ in one measurement basis and $\bar{\bm p}=\frac14\sum_j\bm p_j$. The covariance matrices of the empirical frequencies obey
\begin{equation}
\operatorname{Cov}_{\rm mix}-\operatorname{Cov}_{\rm strat}
 =\frac{1}{n_s}\left(\frac14\sum_{j=1}^{4}\bm p_j\bm p_j^{\mathsf T}
 -\bar{\bm p}\bar{\bm p}^{\mathsf T}\right)\succeq0.
\label{eqS:stratified-covariance}
\end{equation}
Equal-shot pooling is therefore unbiased and has no larger sampling covariance than the mixed-state multinomial simulation. Once $\rho_{\rm occ}$ is reconstructed, $P=4\rho_{\rm occ}$ in the ideal limit.

Using the three-transmon basis in Eq.~\eqref{eqS:ambient-hardware}, the finite-shot protocol employs the $3^3=27$ tensor-product Pauli bases $X,Y,Z$. Figure~\ref{fig:transmon-circuits}(b) shows one member of this circuit family. The measurement stage contains only the local rotations $R_X=H$, $R_Y=HS^\dagger$, and $R_Z=I$, followed by joint computational-basis readout, as demonstrated in the three-qubit superconducting tomography of DiCarlo \textit{et al.} \cite{DiCarlo2010}. Each setting has eight multinomial outcomes and the full set determines all $63$ nontrivial Pauli coefficients. From the estimated coefficients we form the linear-inversion matrix
\begin{equation}
\rho_{\rm LI}=\frac18\sum_{\alpha,\beta,\gamma\in\{I,X,Y,Z\}}
\widehat r_{\alpha\beta\gamma}\,
\sigma_\alpha\otimes\sigma_\beta\otimes\sigma_\gamma,
\label{eqS:pauli-linear-inversion}
\end{equation}
and define $P_{\rm rec}$ as the projector onto its four largest-eigenvalue eigenvectors. The estimator is linear inversion \cite{James2001} followed by rank-four spectral projection; the latter returns the nearest rank-four projector in Frobenius norm. Two additional projector reconstructions suffice if the pole reference subspaces are not pre-calibrated. At each equatorial point, the reconstructed $P$ determines both polar frames and one overlap matrix $G_{NS}$.

We simulate multinomial shot noise with otherwise ideal state preparation, control axes, and readout. Counts are drawn in all $27$ Pauli-basis settings, $P_{\rm rec}$ is reconstructed with Eq.~\eqref{eqS:pauli-linear-inversion}, and the resulting $G_{NS}$ is evaluated with the same $N=4$ single-grid winding estimator used for clean convergence. There are $128$ equatorial points for each bundle and $100$ independent reconstructions per shot budget. If $n_s$ denotes the total shots pooled over the four occupied preparations in one Pauli basis, the budget is $27n_s$ shots per parameter point and $128\times27n_s$ equatorial shots per bundle.

At $n_s=200$, one bundle uses $128\times27=3456$ parameter--basis combinations and $691{,}200$ shots. Preparing the four occupied states separately gives $13{,}824$ circuits at $50$ shots each. The odd bundle and factorized reference together therefore use $27{,}648$ circuits and $1{,}382{,}400$ shots; optional tomography of two pole references adds $10{,}800$ shots per bundle. This is a sampling workload on three qubits, with no entangling gates in the measurement bases. Both residues are identified correctly in all $100$ simulated trials at $n_s=200$, and the mean subspace fidelities are approximately $0.98$ for both bundles.

At this grid resolution the estimator of the integer winding is biased, so we report the parity and its distance from the nearest decision boundary,
\begin{equation}
\widehat\eta=\operatorname{round}(\widehat W_3)\bmod2,
\qquad
m_{\rm parity}=\min_{n\in\mathbb Z}\left|\widehat W_3-\left(n+\tfrac12\right)\right|.
\label{eqS:parity-margin}
\end{equation}
At $n_s=200$, the mean windings are approximately $0.86$ and $1.64$, with mean distances $0.36$ and $0.14$ from the nearest half-integer boundary. The winding distributions and margin statistics are shown in Fig.~\ref{figS:tomography-systematics}(a,b).

Independent symmetric readout errors are added to the three transmon bits. Each measured bit is flipped with probability $r$, while reconstruction deliberately uses the ideal readout model. Figure~\ref{figS:tomography-systematics}(c) therefore gives an unmitigated baseline. A device run would calibrate the joint assignment matrix and can incorporate readout-error mitigation already demonstrated in superconducting state tomography \cite{Aasen2024}. Device-specific coherent, leakage, drift, and correlated-readout errors are outside the present sampling model.

\FloatBarrier
\section*{Data availability}
The numerical data and source code supporting this work are available in the authors' \href{https://github.com/IKEDAKAZUKI/Entangling-Topological-Invariants}{GitHub repository}.

\section*{Acknowledgments}
This work was supported by the NSF under Grant No. OSI-2328774 (KI), and by the Israeli Science Foundation Excellence Center, the US-Israel Binational Science Foundation, and the Israel Ministry of Science (YO).

\bibliography{ref_v4}

@article{WilczekZee1984,
  author = {Wilczek, Frank and Zee, A.},
  title = {Appearance of Gauge Structure in Simple Dynamical Systems},
  journal = {Phys. Rev. Lett.},
  volume = {52},
  pages = {2111--2114},
  year = {1984},
  doi = {10.1103/PhysRevLett.52.2111}
}

@article{Sheng2006,
  author = {Sheng, D. N. and Weng, Z. Y. and Sheng, L. and Haldane, F. D. M.},
  title = {Quantum Spin-{Hall} Effect and Topologically Invariant {Chern} Numbers},
  journal = {Phys. Rev. Lett.},
  volume = {97},
  pages = {036808},
  year = {2006},
  doi = {10.1103/PhysRevLett.97.036808}
}

@article{Prodan2009,
  author = {Prodan, Emil},
  title = {Robustness of the Spin-{Chern} Number},
  journal = {Phys. Rev. B},
  volume = {80},
  pages = {125327},
  year = {2009},
  doi = {10.1103/PhysRevB.80.125327}
}

@misc{Wang2023Feature,
  author = {Wang, Baokai and Hung, Yi-Chun and Zhou, Xiaoting and Ong, Tzen and Lin, Hsin},
  title = {Feature Spectrum Topology and Bulk Boundary Correspondence in Energy and Projective Operator Spectra},
  howpublished = {Commun. Phys.},
  year = {2026},
  note = {published online 30 July 2026},
  doi = {10.1038/s42005-026-02705-5}
}

@misc{Hung2026Nested,
  author = {Hung, Yi-Chun and Ong, T. Tzen and Lin, Hsin},
  title = {Nested Feature Spectrum Topology: Tripartite Topological Equivalence of Feature, Entanglement, and {Wilson} Loop Spectrum},
  year = {2026},
  eprint = {2603.13128},
  archivePrefix = {arXiv},
  primaryClass = {cond-mat.mes-hall}
}

@article{Niu1985,
  author = {Niu, Qian and Thouless, D. J. and Wu, Yong-Shi},
  title = {Quantized {Hall} Conductance as a Topological Invariant},
  journal = {Phys. Rev. B},
  volume = {31},
  pages = {3372--3377},
  year = {1985},
  doi = {10.1103/PhysRevB.31.3372}
}

@article{QiWuZhang2006,
  author = {Qi, Xiao-Liang and Wu, Yong-Shi and Zhang, Shou-Cheng},
  title = {Topological Quantization of the Spin {Hall} Effect in Two-Dimensional Paramagnetic Semiconductors},
  journal = {Phys. Rev. B},
  volume = {74},
  pages = {085308},
  year = {2006},
  doi = {10.1103/PhysRevB.74.085308}
}

@article{Yang1978,
  author = {Yang, C. N.},
  title = {Generalization of {Dirac}'s Monopole to {$\mathrm{SU}_{2}$} Gauge Fields},
  journal = {J. Math. Phys.},
  volume = {19},
  pages = {320--328},
  year = {1978},
  doi = {10.1063/1.523506}
}

@article{Sugawa2018,
  author = {Sugawa, Seiji and Salces-Carcoba, Francisco and Perry, Abigail R. and Yue, Yuchen and Spielman, Ian B.},
  title = {Second {Chern} Number of a Quantum-Simulated non-{Abelian} {Yang} Monopole},
  journal = {Science},
  volume = {360},
  pages = {1429--1434},
  year = {2018},
  doi = {10.1126/science.aam9031}
}

@article{Lohse2018,
  author = {Lohse, Michael and Schweizer, Christian and Price, Hannah M. and Zilberberg, Oded and Bloch, Immanuel},
  title = {Exploring {4D} Quantum {Hall} Physics with a {2D} Topological Charge Pump},
  journal = {Nature},
  volume = {553},
  pages = {55--58},
  year = {2018},
  doi = {10.1038/nature25000}
}

@article{Zilberberg2018,
  author = {Zilberberg, Oded and Huang, Sheng and Guglielmon, Jonathan and Wang, Mohan and Chen, Kevin P. and Kraus, Yaacov E. and Rechtsman, Mikael C.},
  title = {Photonic Topological Boundary Pumping as a Probe of {4D} Quantum {Hall} Physics},
  journal = {Nature},
  volume = {553},
  pages = {59--62},
  year = {2018},
  doi = {10.1038/nature25011}
}

@misc{IkedaEntanglementGeometry,
  author = {Ikeda, Kazuki},
  title = {Quantum Entanglement Geometry on {Severi--Brauer} Schemes: Subsystem Reductions of {Azumaya} Algebras},
  year = {2026},
  eprint = {2601.13764},
  archivePrefix = {arXiv},
  primaryClass = {math.AG}
}

@misc{IkedaOzHolonomy,
  author = {Ikeda, Kazuki and Oz, Yaron},
  title = {Loop-Dependent Entangling Holonomies in Localized Topological Quartets},
  year = {2026},
  eprint = {2604.11596},
  archivePrefix = {arXiv},
  primaryClass = {cond-mat.mes-hall}
}

@misc{Ershov2008,
  author = {Ershov, A. V.},
  title = {Theories of Bundles with Additional Homotopy Conditions},
  year = {2008},
  eprint = {0804.1119},
  archivePrefix = {arXiv},
  primaryClass = {math.KT}
}

@article{Ezawa2014,
  author = {Ezawa, Motohiko},
  title = {Symmetry Protected Topological Charge in Symmetry Broken Phase: Spin-{Chern}, Spin-Valley-{Chern} and Mirror-{Chern} Numbers},
  journal = {Phys. Lett. A},
  volume = {378},
  pages = {1180--1184},
  year = {2014},
  doi = {10.1016/j.physleta.2014.02.014}
}

@article{ZengShengZhu2019,
  author = {Zeng, Tian-Sheng and Sheng, D. N. and Zhu, W.},
  title = {Topological Characterization of Hierarchical Fractional Quantum {Hall} Effects in Topological Flat Bands with {SU(N)} Symmetry},
  journal = {Phys. Rev. B},
  volume = {100},
  pages = {075106},
  year = {2019},
  doi = {10.1103/PhysRevB.100.075106}
}

@article{ZengZhu2022,
  author = {Zeng, Tian-Sheng and Zhu, W.},
  title = {{Chern}-Number Matrix of the non-{Abelian} Spin-Singlet Fractional Quantum {Hall} Effect},
  journal = {Phys. Rev. B},
  volume = {105},
  pages = {125128},
  year = {2022},
  doi = {10.1103/PhysRevB.105.125128}
}

@article{FukuiHatsugaiSuzuki2005,
  author = {Fukui, Takahiro and Hatsugai, Yasuhiro and Suzuki, Hiroshi},
  title = {{Chern} Numbers in Discretized {Brillouin} Zone: Efficient Method of Computing (Spin) {Hall} Conductances},
  journal = {J. Phys. Soc. Jpn.},
  volume = {74},
  pages = {1674--1677},
  year = {2005},
  doi = {10.1143/JPSJ.74.1674}
}

@article{Haldane1995,
  author = {Haldane, F. D. M.},
  title = {Stability of Chiral {Luttinger} Liquids and {Abelian} Quantum {Hall} States},
  journal = {Phys. Rev. Lett.},
  volume = {74},
  pages = {2090--2093},
  year = {1995},
  doi = {10.1103/PhysRevLett.74.2090}
}

@article{Thouless1983,
  author = {Thouless, D. J.},
  title = {Quantization of Particle Transport},
  journal = {Phys. Rev. B},
  volume = {27},
  pages = {6083--6087},
  year = {1983},
  doi = {10.1103/PhysRevB.27.6083}
}

@misc{Ershov2003,
  author = {Ershov, A. V.},
  title = {Homotopy Theory of Bundles with Fiber Matrix Algebra},
  year = {2003},
  eprint = {math/0301151},
  archivePrefix = {arXiv},
  primaryClass = {math.AT}
}

@article{James2001,
  author = {James, Daniel F. V. and Kwiat, Paul G. and Munro, William J. and White, Andrew G.},
  title = {Measurement of Qubits},
  journal = {Phys. Rev. A},
  volume = {64},
  pages = {052312},
  year = {2001},
  doi = {10.1103/PhysRevA.64.052312}
}

@book{MilnorStasheff1974,
  author = {Milnor, John W. and Stasheff, James D.},
  title = {Characteristic Classes},
  series = {Annals of Mathematics Studies},
  number = {76},
  publisher = {Princeton University Press},
  address = {Princeton, NJ},
  year = {1974}
}

@article{Znidaric2008,
  author = {{\v{Z}}nidari{\v{c}}, Marko and Giraud, Olivier and Georgeot, Bertrand},
  title = {Optimal Number of Controlled-{NOT} Gates to Generate a Three-Qubit State},
  journal = {Phys. Rev. A},
  volume = {77},
  pages = {032320},
  year = {2008},
  doi = {10.1103/PhysRevA.77.032320}
}

@article{DiCarlo2010,
  author = {DiCarlo, L. and Reed, M. D. and Sun, L. and Johnson, B. R. and Chow, J. M. and Gambetta, J. M. and Frunzio, L. and Girvin, S. M. and Devoret, M. H. and Schoelkopf, R. J.},
  title = {Preparation and Measurement of Three-Qubit Entanglement in a Superconducting Circuit},
  journal = {Nature},
  volume = {467},
  pages = {574--578},
  year = {2010},
  doi = {10.1038/nature09416}
}

@article{Chow2011,
  author = {Chow, J. M. and C{\'o}rcoles, A. D. and Gambetta, J. M. and Rigetti, C. and Johnson, B. R. and Smolin, J. A. and Rozen, J. R. and Keefe, G. A. and Rothwell, M. B. and Ketchen, M. B. and Steffen, M.},
  title = {Simple All-Microwave Entangling Gate for Fixed-Frequency Superconducting Qubits},
  journal = {Phys. Rev. Lett.},
  volume = {107},
  pages = {080502},
  year = {2011},
  doi = {10.1103/PhysRevLett.107.080502}
}

@article{Wei2024,
  author = {Wei, Ken Xuan and Lauer, Isaac and Pritchett, Emily and Shanks, William and McKay, David C. and Javadi-Abhari, Ali},
  title = {Native Two-Qubit Gates in Fixed-Coupling, Fixed-Frequency Transmons Beyond Cross-Resonance Interaction},
  journal = {PRX Quantum},
  volume = {5},
  pages = {020338},
  year = {2024},
  doi = {10.1103/PRXQuantum.5.020338}
}

@article{Aasen2024,
  author = {Aasen, Adrian Skasberg and Di Giovanni, Andras and Rotzinger, Hannes and Ustinov, Alexey V. and G{\"a}rttner, Martin},
  title = {Readout Error Mitigated Quantum State Tomography Tested on Superconducting Qubits},
  journal = {Commun. Phys.},
  volume = {7},
  pages = {301},
  year = {2024},
  doi = {10.1038/s42005-024-01790-8}
}

@article{Mottonen2005,
  author = {M{\"o}tt{\"o}nen, Mikko and Vartiainen, Juha J. and Bergholm, Ville and Salomaa, Martti M.},
  title = {Transformation of Quantum States Using Uniformly Controlled Rotations},
  journal = {Quantum Inf. Comput.},
  volume = {5},
  number = {6},
  pages = {467--473},
  year = {2005},
  doi = {10.26421/QIC5.6-5}
}

\end{document}